\documentclass{amsart}
\usepackage[letterpaper,hmargin=1.1in,vmargin=1.57in]{geometry}

\usepackage{amsmath, amsthm, amssymb}
\usepackage{enumerate}

\newtheorem{theorem}{Theorem}[section]
\newtheorem{lemma}[theorem]{Lemma}
\newtheorem{algorithm}[theorem]{Algorithm}
\newtheorem*{algorithmoutline1}{Drinfeld Module Analogue of the Black-box Berlekamp Algorithm}

\theoremstyle{definition}
\newtheorem{definition}[theorem]{Definition}
\newtheorem{example}[theorem]{Example}

\newtheorem{conjecture}[theorem]{Conjecture}
\newtheorem{corollary}[theorem]{Corollary}

\theoremstyle{remark}
\newtheorem{remark}[theorem]{Remark}

\numberwithin{equation}{section}

\newcommand{\Z}{\mathbb Z}
\newcommand{\F}{\mathbb F}

\newcommand{\A}{\mathbb A}

\newcommand{\D}{\Delta}

\newcommand{\p}{\mathfrak p}
\newcommand{\h}{\mathfrak h}
\newcommand{\f}{\mathfrak f}

\newcommand{\B}{\mathfrak B}
\newcommand{\ef}{\mathcal F}

\usepackage{tikz}
\usetikzlibrary{matrix,arrows}

\newcommand{\ph}{(\phi/\p)}

\newcommand{\g}{\mathfrak g}

\title[Polynomial Factorization And Euler-Poincare Characteristics of Drinfeld Modules]{Polynomial Factorization over Finite Fields By Computing Euler-Poincare Characteristics of Drinfeld Modules}
\author[Anand Kumar Narayanan]{Anand Kumar Narayanan\\~\\\MakeLowercase{anandkn@caltech.edu}} 
\thanks{The author was partially supported by NSF grant CCF 1423544}
\address{Anand Kumar Narayanan, Department of Computing and Mathematical Sciences, California Institute of Technology, Pasadena, California, USA, 91125}
\email{anandkn@caltech.edu}

\begin{document}
\maketitle
\begin{abstract}
We propose and rigorously analyze two randomized algorithms to factor univariate polynomials over finite fields using rank $2$ Drinfeld modules. The first algorithm estimates the degree of an irreducible factor of a polynomial from Euler-Poincare characteristics of random Drinfeld modules. Knowledge of a factor degree allows one to rapidly extract all factors of that degree. As a consequence, the problem of factoring polynomials over finite fields in time nearly linear in the degree is reduced to finding Euler-Poincare characteristics of random Drinfeld modules with high probability.
The second algorithm is a random Drinfeld module analogue of Berlekamp's algorithm. During the course of its analysis, we prove a new bound on degree distributions in factorization patterns of polynomials over finite fields in certain short intervals.
\end{abstract}
\section{Introduction}
\subsection{Current State of the Art}
Let $\F_q$ denote the finite field with $q$ elements and $\F_q[t]$ the polynomial ring in one indeterminate. The fastest known randomized algorithm for factorization in $\F_q[t]$ is the Kaltofen-Shoup algorithm \cite[\S~2]{ks} implemented by Kedlaya-Umans fast modular composition \cite{ku}. It belongs in the Cantor-Zassenhaus \cite{cz} framework and to factor a polynomial of degree $n$ takes $\widetilde{\mathcal{O}}(n^{3/2} \log q + n \log^2 q)$ expected time by employing the following sequence of steps\footnote{The soft $\widetilde{\mathcal{O}}$ notation suppresses $n^{o(1)}$ and $\log^{o(1)} q$ terms for ease of exposition.}. The first is square free factorization where the polynomial in question is written as a product of square free polynomials. A square free polynomial is one that does not contain a square of an irreducible polynomial as a factor. The second step known as distinct degree factorization takes a monic square free polynomial and decomposes it into factors each of which is a product of irreducible polynomials of the same degree. The final step is equal degree factorization which splits a polynomial all of whose irreducible factors are of the same degree into irreducible factors. The bottleneck is distinct degree factorization and currently the difficulty appears to be in finding the factor degrees. Given the degree of a factor, one can extract all factors of that degree in $\widetilde{\mathcal{O}}( n \log^2 q)$ expected time \cite{ku}.
\subsection{Polynomial Factorization and Euler-Poincare Characteristic of Drinfeld Modules}
We propose a novel procedure to read off the smallest factor degree of a monic square free polynomial $h \in \F_q[t]$ from the Euler-Poincare characteristic $\chi_{\phi,h}$ of a random rank 2 Drinfeld module $\phi$ reduced at $h$. The reduction of $\phi$ at $h$ is a finite $\F_q[t]$-module and its Euler-Poincare characteristic $\chi_{\phi,h} \in \F_q[t]$ (see Definition \ref{euler}) is an $\F_q[t]$ valued cardinality measure. 
If $h$ factors into monic irreducible factors as $h = \prod_{i} p_i$, then $\chi_{\phi,h} = \prod_i \chi_{\phi,p_i}$. A Drinfeld module analogue of Hasse's theorem for elliptic curves, due to Gekeler \cite{gek1}, asserts for each $p_i$ that $\chi_{\phi,p_i} = p_i + c_{\phi,i}$ for some $c_{\phi,i} \in \F_q[t]$ with $\deg(c_{\phi,i}) \leq \deg(p_i)/2$. Hence, the leading coefficients of $h$ and $\chi_{\phi,h}$ agree and the number of agreements reveals information about the degree of the smallest degree factor of $h$.  In \S~\ref{degree_guessing_section}, we prove,  
\begin{theorem}\label{chi_reduction_theorem}
For $n \leq \sqrt{q}/2$, the smallest factor degree of a degree $n$ square free $h \in \F_q[t]$ can be inferred in $O(n \log q)$ time from $\chi_{\phi,h}$ with probability at least $1/4$ for a randomly chosen $\phi$.
\end{theorem}
\noindent Consider an algorithm $\mathcal{B}$ that takes a square free $h \in \F_q[t]$ and a random Drinfeld module $\phi$ as inputs such that with constant probability the output $\mathcal{B}(h,\phi)$ is $\chi_{\phi,h}$. That is, $B$ is a montecarlo algorithm to compute Euler-Poincare characteristics. By choosing $\phi$ at random and invoking Theorem \ref{chi_reduction_theorem}, we establish that a non trivial factor can be found in nearly linear time with oracle access  to $\mathcal{B}$.
\begin{corollary}\label{chi_factor_corollary}
\textit{There exists an $\widetilde{\mathcal{O}}(n  \log^2 q)$ expected time algorithm (with oracle access to $\mathcal{B}$) to find an irreducible factor of a square free polynomial $h \in \F_q[t]$ of degree $n \leq \sqrt{q}/2$ with only $O(1)$ queries to $\mathcal{B}$, each of the form $\mathcal{B}(h)$.}
\end{corollary}
\noindent The requirement $\sqrt{q} \geq 2n$ in Theorem \ref{chi_reduction_theorem} and Corollary \ref{chi_factor_corollary} is without loss of generality since if $q$ were smaller, we could choose a small power $q^\prime$ of $q$ that satisfies $\sqrt{q^\prime} \geq 2n$ and obtain the factorization over $\F_{q^\prime}[t]$ with the running time unchanged up to polylogarithmic factors in $n$ (see Remark \ref{large_q}).\\ \\
Given oracle access to $\mathcal{B}$, obtaining the complete factorization by naively extracting one factor at a time using Corollary \ref{chi_factor_corollary} leads to a $3/2$ running time exponent on $n$. For instance, a polynomial of degree $n$ with one irreducible factor each of degree $1,2,3,\ldots,m-1,m$ where $m = \Theta(n^{1/2})$ requires $\Theta(n^{1/2})$ extractions.  In \cite{gnu}, an algorithm to quickly obtain the complete factorization given a procedure to extract a non trivial factor as a subroutine is described. As a consequence, an algorithm for implementing $\mathcal{B}$ with exponent less than $3/2$ would lead to a polynomial factorization algorithm with exponent less than $3/2$. As an illustration, in \S~\ref{complete_factorization}, we obtain the following corollary of Theorem \ref{chi_reduction_theorem}  describing an implication of a nearly linear time algorithm for computing Euler-Poincare characteristic of random Drinfeld modules.
\begin{corollary}\label{chi_reduction_corollary}
\textit{An implementation of the oracle function $\mathcal{B}$ that takes $\widetilde{\mathcal{O}}(n \log^{O(1)} q)$ expected time for inputs of degree $n$ yields an $\widetilde{\mathcal{O}}(n^{4/3} \log^{O(1)} q)$ expected time polynomial factorization algorithm.}
\end{corollary}
\noindent We describe our first complete 
polynomial factorization algorithm by presenting an implementation of $\mathcal{B}$. Assuming the matrix multiplication exponent $\omega$ is $2$, our implementation of $\mathcal{B}$ and hence the polynomial factorization algorithm both have running time exponent $3/2$ in the input degree. A faster implementation of $\mathcal{B}$ would break the $3/2$ exponent barrier in polynomial factorization. Thus, the problem of computing Euler-Poincare characteristics of random Drinfeld modules warrants a thorough investigation. The problem is analogous to point counting on elliptic curves over finite fields and the question as to if there is a Drinfeld module analogue of Schoof's algorithm \cite{sch} is immediate.\\ \\
We next briefly sketch our implementation of $\mathcal{B}$. The obvious procedure to compute $\chi_{\phi,h}$ for a given a square free $h \in A$ and a rank $2$ Drinfeld module $\phi$, is to compute the characteristic polynomial of the ($\F_q$ linear) $\phi$ action on $\F_q[t]/(h)$. However, the complexity of such generic linear algebraic techniques is equivalent to inverting square matrices of dimension $\deg(h)$. 
We devise a faster implementation of $\mathcal{B}$ by exploiting the fact that the input $\phi$ to $\mathcal{B}$ is chosen at random. For $q \geq 2\deg(h)^4$, which we may assume without loss of generality, we prove it is likely that the reduction of a random $\phi$ at $h$ is a cyclic $\F_q[t]$-module and further that $\chi_{\phi,h}$ coincides with the order (that is, the monic generator of the annihilator) of a random element in $\phi$ reduced at $h$. Implementing $\mathcal{B}$ is thus reduced to the finding the order of a random element in a random Drinfeld module $\phi$ reduced at $h$ with constant probability. 
In \S~\ref{order_computation}, we solve the order finding problem by posing it as an instance of the automorphism projection problem of Kaltofen-Shoup \cite{ks} and thereby obtain an $\widetilde{\mathcal{O}}(n^{(1+\omega)/2} \log q + n \log^2 q)$ expected time algorithm to find a factor. In \S~\ref{complete_factorization}, we describe how to obtain the complete factorization.
\subsection{Drinfeld Module Analog of Berlekemp's Algorithm}
Our second algorithm is a randomized Drinfeld analogue of Berlekamp's algorithm \cite{ber} wherein $\F_q[t]$-modules twisted by the Frobenius action is replaced with reductions of random rank-2 Drinfeld modules. It has the distinction of being the only  polynomial factorization algorithm over finite fields that does not involve a quadratic residuosity like map. To factor $h \in \F_q[t]$, Berlekamp's algorithm proceeds by finding a basis for the  Berlekamp subalgebra, which is the fixed space of the $q^{th}$ power Frobenius $\tau$ acting on $\F_q[t]/(h)$. Then (as in the Cantor-Zassenhaus \cite{cz} variant) a random element $\beta$ in the Berlekamp subalgebra is generated by taking a random $\F_q$ linear combination of the basis elements. When $q$ is odd, with probability at least half, the greatest common divisor of $h$ and a lift of $\beta^{(q-1)/2}-1$ is a non trivial factor of $h$. In place of the $\F_q[t]$-module $\F_q[t]/(h)$ and its Berlekamp subalgebra (which is the $\tau-1$ torsion in $\F_q[t]/(h)$), our second algorithm works over a random Drinfeld module reduced at $h$ and its torsion corresponding to low degree polynomials in the Drinfeld action. These low degree polynomials are precisely the low degree factors of $\chi_{\phi,h}$. In terms of implementation, the algorithm closely resembles the fast black-box algorithm of Kaltofen-Shoup \cite[\S~3]{ks} and shares its $\widetilde{\mathcal{O}}(n^{(1+\omega)/2+o(1)} \log q + n \log^2 q)$ expected running time.\\ \\
Our analysis of the Drinfeld analogue of Berlekamp's algorithm relies on bounds on the degree distribution of factorization patterns of polynomials in short intervals in $\F_q[t]$. We prove the required bounds, but only when $q$ is very large compared to $n$. When $q$ is not large enough, the claimed running times hold under a widely believed conjecture (see Conjecture \ref{smoothness_conjecture}) and a slower running time bound is proven unconditionally (see Remark \ref{rem2}). We next state the bounds for large $q$ since they might be of independent interest.
\subsection{Factorization Patterns of Polynomials in Short Intervals}
For a partition $\lambda$ of a positive integer $d$, let $P(\lambda)$ denote the fraction of permutations on $d$ letters whose cycle decompositions correspond to $\lambda$. When $q$ is large enough compared to $d$, a random polynomial in $\F_q[t]$ of degree $d$ has a factorization pattern corresponding to a partition $\lambda$ of $d$ with probability about $P(\lambda)$ \cite{coh}. In  \S~\ref{factor_distribution}, we prove that the degree distribution of a random polynomial in the interval $\mathcal{I}_{f,m}:=\{f+a | a\in \F_q[t], \deg(a) \leq m\}$ around $f \in \F_q[t]$ is not far from the degree distribution of a random polynomial of degree $d$.
\begin{theorem}\label{small_interval_bound_intro}  
For every $f \in \F_q[t]$ of degree $d$ bounded by $\log q \geq 3 d \log d$, for every $m \geq 2$ and for every partition $\lambda$ of $d$, 
$$ \left(1-\frac{1}{\sqrt{q}}\right) P(\lambda) \leq  \frac{\left|\{g \in \mathcal{I}_{f,m} |  \lambda_{g} = \lambda\}\right|}{|\mathcal{I}_{f,m}|} \leq \left(1+\frac{1}{\sqrt{q}}\right) P(\lambda)$$ 
where  $\lambda_g$ denotes the partition of $\deg(g)$ induced by the degrees of the irreducible factors of $g$.
\end{theorem}
\subsection{Related Work}
Our algorithms draw inspiration from Lenstra's elliptic curve integer factorization \cite{len} where the role of multiplicative groups modulo primes in Pollard's $p-1$ algorithm \cite{pol} was recast with the group of rational points on random Elliptic curves modulo primes. Our first algorithm for degree estimation using Euler-Poincare characteristic is a random Drinfeld module analogue of an algorithm described in the author's Ph.D thesis \cite[Chap 7]{nar} using Carlitz modules. To our knowledge, the use of Drinfeld modules to factor polynomials over finite fields originated with Panchishkin and Potemine \cite{pp} whose algorithm was rediscovered by van der Heiden \cite{vdH} (see also \cite{vdH1}). Our Drinfeld module analogue of Berlekamp's algorithm shares some similarities with the algorithm in \cite{vdH} but they differ in the following aspects. In contrast to our algorithm, the algorithm in \cite{vdH} only works for equal degree factorization and targets torsion corresponding to large degree polynomials to aid in the splitting; resulting in a slower running time. Further, its analysis was merely supported by heuristics. Using  Theorem \ref{small_interval_bound_intro} and Lemma \ref{set_bound2}, the proof of Lemma \ref{alg2_proof} can be extended to rigorously analyze the algorithm in \cite{vdH} for large $q$ when restricted to rank $2$ Drinfeld modules.
\subsection{Organization}
The analysis of our algorithms rely critically on the distribution of the characteristic polynomial (of the Frobenius endomorphism in the representation of the endomorphism ring) of a random rank $2$ Drinfeld module (on its $\ell$-adic Tate modules). In \S~\ref{drinfeld_notation}, we prove Lemma \ref{set_bound2}, the required equidistribution lemma on the aforementioned characteristic polynomial. We begin \S~\ref{drinfeld_notation} by recounting the theory of rank $2$ Drinfeld modules followed by \S~\ref{drinfeld_structure}, where the structure of rank $2$ Drinfeld modules in terms of its characteristic polynomial is described.  In \S~\ref{frobenius_distributions}, we state a weighted class number formula for isomorphism classes of Drinfeld modules with a given characteristic polynomial (due to Gekeler \cite{gek}) and from it derive Lemma \ref{set_bound2}. Gekeler's weighted class number formula is obtained through complex multiplication theory, that is, a correspondence between isomorphism classes of Drinfeld modules and Gauss class numbers in certain imaginary quadratic orders. While they hold in even characteristic, for ease of exposition, we assume in \S~\ref{frobenius_distributions} and consequently in the analysis of our algorithms that $q$ is odd. In \S~\ref{degree_guessing_section}, we state and analyze our first algorithm, namely to estimate factor degrees by computing Euler-Poincare characteristic and prove Theorem \ref{chi_reduction_theorem} and Corollaries \ref{chi_factor_corollary} and \ref{chi_reduction_corollary}. 
Theorem \ref{small_interval_bound_intro} concerning factorization patterns of polynomials in short intervals is proven in \S~\ref{factor_distribution}, which can be read independent of the rest of the paper. The Drinfeld analogue of Berlekamp's algorithm is described and analyzed in \S~\ref{blackbox_section}. \\ \\
A remarkable feature of our algorithms and reductions is that we address the factorization of a square free $h = \prod_i p_i \in \F_q[t]$ by looking at $\chi_{\phi,h} = \prod_i \chi_{\phi,p_i}$ for random $\phi$. We may view $\chi_{\phi,h}$ as a perturbation of $h$, since $\chi_{\phi,p_i}$ is the irreducible factor $p_i$ of $h$ perturbed within a half degree (of $p_i$) interval. These intervals widen with Drinfeld modules of increasing rank. Theorem \ref{small_interval_bound_intro} and Lemma \ref{set_bound2} together imply that even for rank $2$, the intervals are large enough for $\chi_{\phi,p_i}$ to exhibit random factorization patterns. This allowed us to pose the worst case complexity of polynomial factorization in terms of an average complexity statement (Corollary \ref{chi_reduction_corollary}) concerning Drinfeld modules. It would be interesting to see if this perturbation leads to worst case to average case reductions for other multiplicative problems.
\section{Finite Rank-2 Drinfeld Modules}\label{drinfeld_notation}
\noindent Let $A=\F_q[t]$ denote the polynomial ring in the indeterminate $t$ and let $K$ be a field with a non zero ring homomorphism $\gamma:A \rightarrow K$. Necessarily, $K$ contains $\F_q$ as a subfield. Fix an algebraic closure $\bar{K}$ of $K$ and let $\tau: \bar{K} \longrightarrow \bar{K}$ denote the $q^{th}$ power Frobenius endomorphism. The ring of endomorphisms of the additive group scheme $\mathbb{G}_a$ over $K$ can be identified with the skew polynomial ring $K\langle \tau \rangle$ where $\tau$ satisfies the commutation rule $\forall u \in K, \tau u = u^q \tau$.
A rank-2 Drinfeld module over $K$ is (the $A$-module structure on $\mathbb{G}_a$ given by) a ring homomorphism
$$\phi : A \longrightarrow K\langle \tau \rangle$$ 
$$\ \ \ \ \ \ \ \ \ \ \ \ \ \ \ \ \ \ \ \ \ \ t \longmapsto \gamma(t) + g_\phi \tau + \D_\phi \tau^2$$
for some $g_\phi \in K$ and $\D_\phi \in K^\times$. For $a \in A$, let $\phi_a$ denote the image of $a$ under $\phi$. We will concern ourselves primarily with rank $2$ Drinfeld modules and unless otherwise noted, a Drinfeld module will mean a rank $2$ Drinfeld module.\\ \\
To every $A$-algebra $L$ over $\bar{K}$, the Drinfeld module $\phi$ endows a new $A$-module structure (which, we denote by $\phi(L)$) through the $A$-action $$\forall f \in L, \forall a \in A, a*f = \phi_a(f).$$  
For every $A$-algebra homomorphism $\rho:L \longrightarrow L$,  $\forall a\in A$ and $\forall f \in L$, $\rho(\phi(f)) = \phi(\rho(f))$. Thus $\rho$, when thought of as a map from $\phi(L) \longrightarrow \phi(L)$, is an $A$-module homomorphism. For every direct product $L \times L^\prime$ of $A$-algebras over $\bar{K}$, we hence have the corresponding direct sum of $A$-modules $$\phi(L \times L^\prime) \cong \phi(L) \oplus \phi(L^\prime).$$
Henceforth, we restrict our attention to Drinfeld modules $\phi:A \longrightarrow \F_q(t)$ over $\F_q(t)$ (with $\gamma : A \rightarrow \F_q(t)$ being the inclusion (identity map), $g_\phi\in A$ and $\Delta_\phi\in A^\times$) and their reductions.\\ \\
For a proper ideal $\f \subset A$, let $\F_\f$ denote $A/\f$. For a prime ideal $\p \subset A$, if $\D_\phi$ is non zero modulo $\p$, then the reduction $\phi/\p := \phi \otimes \F_\p$ of $\phi$ at $\p$ is defined through the ring homomorphism $$\phi/\p : A \longrightarrow \F_\p\langle \tau \rangle$$ $$\ \ \ \ \ \ \ \ \ \ \ \ \ \ \ \ \ \ \ \ \ \ \ \ \ \ \ \ \ \ \ \ \ \ \ \ \ \ \ \ \ \ \ \ \ \ \ \ t \longmapsto t + (g_\phi \mod \p) \tau + (\D_\phi\mod\p) \tau^2$$
and the image of $a \in A$ under $\phi/\p$ is denoted by $(\phi/\p)_a$. Even if $\Delta_{\phi}$ is zero modulo $\p$, one can still obtain the reduction $(\phi/\p)$ of $\phi$ at $\p$ through minimal models of $\phi$ (c.f. \cite{gek1}). We refrain from addressing this case since our algorithms do not require it.\\ \\
As before, the Drinfeld module $\phi/\p$ endows a new $A$-module structure (denoted by $\ph(L)$) to every $A$-algebra $L$ over the algebraic closure of $\F_\p$ through the $A$-action $$\forall f \in L, \forall a \in A, a*f = \ph_a(f)$$ and for every direct product $L \times L^\prime$ of $A$-algebras over the algebraic closure of $\F_\p$, $$\ph(L \times L^\prime) \cong \ph(L) \oplus \ph(L^\prime).$$
Further, for every $A$-algebra $L$ over the algebraic closure of $\F_\p$, $\phi(L) \cong \ph(L)$.\\ \\
Define the annihilator $Ann(L)$ of a finite $A$-module $L$ to be the monic generator of the annihilator ideal $\{a \in A | aL = 0 \}$ of $L$. Define the $A$-order $Ord(\alpha)$ of an element $\alpha$ in a finite $A$-module $L$ to be the monic generator of the annihilator ideal $\{a \in A | a\alpha =0 \}$ of $\alpha$.
For $f \in A$, denote by $(f)$ the ideal generated by $f$ and by $\deg(f)$ the degree of $f$. For a non zero ideal $\f \subset A$, let $\deg(\f)$ denote the degree of its monic generator. For $f,g \in A$, by $\gcd(f,g)$ we mean the monic generator of the ideal generated by $f$ and $g$.
\begin{definition}\label{euler}
Following \cite{gek1}, the Euler-Poincare characteristic of a finite $A$-module $L$ is defined as the unique monic polynomial $\chi(L) \in A$ such that
\begin{enumerate}
\item If $L \cong A/\p$ for a prime ideal $\p \subseteq A$, then $(\chi(L)) = \p$,
\item If $0\rightarrow L_1 \rightarrow L \rightarrow L_2 \rightarrow 0$ is exact, then $\chi(L)=\chi(L_1)\chi(L_2)$.
\end{enumerate}
\end{definition}
\noindent The above definition is a minor departure from \cite{gek1}, where $\chi(L)$ was defined as an ideal, not a monic generator. For a non zero ideal $\f \subset A$ and a Drinfeld module $\phi$, let $\chi_{\phi,\f}$ denote $\chi(\phi(\F_\f))$. By definition, $Ann(\phi(\F_\f))$ divides $\chi_{\phi,\f}$.
\subsection{Frobenius Distributions and the Structure of Rank 2 Drinfeld Modules}\label{drinfeld_structure}
In this subsection, we recount the characterization of the $A$-module structure of reductions of Drinfeld modules at primes due to Gekeler. Unless stated otherwise, proofs of claims made here are in \cite{gek1}. For Drinfeld modules $\phi,\psi$ with reduction at a prime $\p \subset A$, a $\mu \in \F_\p\langle \tau \rangle$ such that $$\forall a \in A,\ \ \mu (\phi/\p)_a = (\psi/\p)_a \mu$$ is called as a morphism from $\phi/\p$ to $\psi/\p$. Let $End_{\F_\p}(\phi)$ denote the endomorphism ring of $\phi/\p$. The Frobenius at $\p$, $\tau^{\deg(\p)}$ is in $End_{\F_\p}(\phi)$ and there exists a polynomial $$P_{\phi,\p}(X) = X^2-a_{\phi,\p}X+b_{\phi,\p} \in A[X]$$ (called the characteristic polynomial of the Frobenius at $\p$) such that $P_{\phi,\p}(\tau^{\deg(\p)}) = 0$ in $End_{\F_\p}(\phi)$. The polynomial $P_{\phi,\p}$ is called as the characteristic polynomial because it is the characteristic polynomial of  $\tau^{\deg(\p)}$ in the representations of $End_{\F_\p}(\phi)$ on the $\ell$-adic Tate modules for prime ideals $\ell \subset A$. \\ \\ 
Furthermore, a rank-2 Drinfeld module analogue of Hasse's theorem for elliptic curves states that $$(b_{\phi,\p}) = \p\  , \ \deg(a_{\phi,\p}) \leq \deg(\p)/2.$$ To be precise, the coefficient  $b_{\phi,\p} = \epsilon_{\phi,\p} p $ where $p$ is the monic generator of $\p$ and $$\epsilon_{\phi,\p} =1/((-1)^{\deg{\p}} \mathcal{N}_{\F_\p/\F_q}(\Delta_\phi)) \in \F_q^\times$$ where $\mathcal{N}_{\F_\p/\F_q}$ is the norm from $\F_\p$ to $\F_q$. As in the elliptic curve case, $a_{\phi,\p}$ is referred to as the Frobenius trace; a consequence of the aforementioned connection to Tate modules. Since $\tau^{\deg(\p)}$ acts as the identity on $\phi(\F_\p)$, $P_{\phi,\p}(1)$  kills $\phi(\F_p)$ (that is, $\ph_{P_{\phi,\p}(1)}$ is the zero element in $End_{\F_\p}(\phi)$). In fact, $$(\chi_{\phi,\p}) = (P_{\phi,\p}(1))  \Rightarrow \chi_{\phi,\p} = p - (a_{\phi,\p} - 1)/\epsilon_{\phi,\p}.$$
As a consequence of $\phi$ being of rank $2$, $\phi(\F_\p)$ is either a cyclic $A$-module or a direct sum of two cyclic $A$-modules. That is, there exists monic polynomials $m_{\phi,\p},n_{\phi,\p} \in A$ (not necessarily relatively prime) such that as $A$-modules $$ \phi(\F_\p) \cong  A/(m_{\phi,\p}) \oplus A/(m_{\phi,\p}n_{\phi,\p}).$$
In particular, $$\chi_{\phi,\p} = p - (a_{\phi,\p} - 1)/\epsilon_{\phi,\p} = m_{\phi,\p}^2n_{\phi,\p}\ ,\ Ann(\phi(\F_\p)) = lcm(m_{\phi,\p}, m_{\phi,\p}n_{\phi,\p}) = m_{\phi,\p}n_{\phi,\p}.$$
Further still, when $q$ is odd, $m_{\phi,\p}$ and $n_{\phi,\p}$ are completely determined by $P_{\phi,\p}$ and a precise characterization of $m_{\phi,\p}$ and $n_{\phi,\p}$ in terms of $a_{\phi,\p}$ and $\epsilon_{\phi,\p}$ is described by Cojocaru and Papikian \cite[Cor3]{cp}.\\ \\ 
Seeing that the characteristic polynomial completely determines the $A$-module structure of a Drinfeld module reduced at a prime $\p$, the question as to which polynomials can arise as such characteristic polynomials is immediate and is addressed in \cite{dri1}, \cite{gek1} and \cite{yu}. Yu \cite{yu} relates the number of isomorphism classes of Drinfeld modules $(\phi/\p)$ with a given characteristic polynomial $X^2-aX+\epsilon p$ to class numbers of orders in imaginary quadratic extensions over $K$. This relation arises from a theory that bears likeness to the complex multiplication theory of elliptic curves. 
Gekeler \cite{gek} using the complex multiplication theory proved a precise characterization of the number of isomorphism classes of Drinfeld modules $(\phi/\p)$ with a given characteristic polynomial $X^2-aX+\epsilon p$. Gekeler's characterization implies a certain equidistribution of the probability of $X^2-aX+\epsilon p$ being $P_{\phi,\p}$ for a randomly chosen $\phi$ and is discussed in the subsequent subsection. 
\subsection{Frobenius Distributions of Rank 2 Drinfeld Modules}\label{frobenius_distributions}
\noindent Let $\p \subset A$ be a prime of degree $d$ and let $p$ be its monic generator. Analysis of our algorithms will involve counting Drinfeld modules $(\phi/\p)$ with characteristic polynomial $P_{\phi,\p}(X) = X^2-aX+\epsilon p$ for a given $\epsilon \in \F_q^\times$ and $a \in A$ of degree at most $d/2$. Such precise counts were proven by Gekeler \cite{gek} by building on the connection between isomorphism classes of Drinfeld modules over $\F_\p$ and  class numbers of imaginary quadratic orders established by Yu \cite{yu}. We begin the section by stating the result of Gekeler in equation \ref{hstar} and later prove lemmas that will find use in the analysis of our algorithms.\\ \\
We identify a Drinfeld module $(\phi/\p)$ with the tuple $(\g_\phi \mod \p, \Delta_\phi \mod \p)$. The number of Drinfeld modules $(\phi/\p)$ is $|\F_\p||\F_\p^\times|$ since we get to pick $g_\phi \mod \p$ from $\F_\p$ and $\Delta_\phi \mod \p$ from $\F_\p^\times$. Two Drinfeld modules $(\phi/\p)$ and $(\psi/\p)$ are isomorphic over $\F_\p$ if and only if there is a $c \in \F_\p^\times$ such that $g_\psi = c^{q-1}g_\phi \mod \p$ and $\Delta_\psi = c^{q^2-1}\Delta_\phi \mod \p$.\\ \\
As a consequence, the automorphism group $Aut_{\F_\p}(\phi)$ depends on if $g_\phi \mod \p = 0$. If $g_\phi \mod \p = 0$ and if further $\F_\p$ contains a quadratic extension (call $\F_{q^2}$) of $\F_q$, then $Aut_{\F_\p}(\phi) \cong \F_{q^2}^\times$. Else, $Aut_{\F_\p}(\phi) \cong \F_q^\times$. The former case $Aut_{\F_\p}(\phi) \cong \F_{q^2}^\times$ corresponds to $\phi$ having complex multiplication by $\F_{q^2}[t]$ and is rare.\\ \\
For $\epsilon \in \F_q^\times$ and $a \in A$ not necessarily monic and of degree at most $\deg(\p)/2$, denote by $H(a,\epsilon,\p)$ a set of representatives of isomorphism classes of Drinfeld modules $(\phi/\p)$ with $P_{\phi,\p}(X) = X^2-aX+\epsilon p$. Define $$h^*(a,\epsilon,\p):= \sum_{(\phi/\p) \in H(a,\epsilon,\p)} \frac{q-1}{\#Aut_{\F_\p}(\phi)}$$
which might be thought of as a weighted count of the isomorphism classes of Drinfeld modules with $P_{\phi,\p}(X) = X^2-aX+\epsilon p$ where the weight $(q-1)/Aut_{\F_\p}(\phi)$ is $1$ except in rare cases.
We next describe the connection between class numbers of certain imaginary quadratic orders and $h^*(a,\epsilon,\p)$.\\ \\ 
Fix an $\epsilon \in \F_q^\times$ and an $a \in A$ of degree at most $\deg(\p)/2$. Let $C$ be the $A$-algebra generated by a root of $X^2-aX+\epsilon p$ and let $E$ be the quotient field of $C$. It turns out that $E$ is an imaginary quadratic extension of $k$.\\ \\
Yu \cite{yu} proved that two Drinfeld modules are isogenous if and only if they have the same characteristic polynomial and further established that the number of isomorphism classes of Drinfeld modules with characteristic polynomial $X^2-aX+\epsilon p$ equals the Gauss class number of $C$. This connection is analogous to a similar statement concerning elliptic curves due to Deuring \cite{deu}. Further, the weighted count $h^*(a,\epsilon,\p)$ equals a certain appropriately weighted Gauss class number of $C$ which was explicitly computed by Gekeler \cite{gek} using an analytic class number formula. We next summarize this result of Gekeler assuming for ease of exposition that $\F_q$ is of odd characteristic throughout this section.\\ \\
Let $B$ be the integral closure of $A$ in $E$. Let $D$ denote the discriminant $a^2-4\epsilon p$ and $f$ the largest monic square factor of $D$. Let $D_0 = D/f$. Then $C = A + f B$.\\ \\
Let $\xi$ denote the Dirichlet character associated with $E$ and for $\Re(s)>1$ define the L-function $$L(s,\chi) := \prod_{\ell\ prime\ of\ k}\left(1- \xi(\ell) |\ell|^{-s} \right)^{-1}.$$ 
The unique prime at infinity $\infty$ is ramified in $E/k$ if $\deg(D_0)$ is odd and is inert if $\deg(D_0)$ is even. Let $\eta$ denote the ramification index of $\infty$ in $E/k$ (that is, $\eta = 2$ if $\deg(D_0)$ is odd and $\eta=1$ otherwise). Let $g$ denote the genus of the algebraic curve associated with $E$. Then $$h^*(a,\epsilon,\p) = \eta q^g S(f) L(1,\xi)$$ where $$S(f):= \sum_{\f^\prime | \f} |\F_{\f^\prime}| \prod_{\ell | \f^\prime} (1-\xi(\ell)|\F_\ell|^{-1}).$$
Here $\f \subseteq A$ is the ideal generated by $f$, the summation is over proper ideals $\f^\prime$ dividing $\f$ and the product is over prime ideals $\ell$ dividing $\f^\prime$. When $\f$ is not $A$, we have $$S(f) \geq |\F_\f|.$$
The conductor $cond(\xi)$ of $\xi$ is 
\begin{equation*}
cond(\xi) = \begin{cases}
(D_0) &\text{if $\deg(D_0)$ is even}\\
(D_0).\infty &\text{if $\deg(D_0)$ is odd}
\end{cases}
\end{equation*}
But for a couple of exceptional cases, the genus $g$ is determined by $cond(\xi)$ as $$g=\deg(cond(\xi))/2-1.$$ The inequality \ref{hstar} we arrive at for $h^*(a,\epsilon,\p)$ will be accurate in those exceptional cases as well. Thus we refrain from mentioning the exceptional cases referring instead the interested reader to \cite{gek}. Thus 
\begin{equation*}
q^g = \begin{cases}
q^{\deg(D_0)/2-1} &\text{if $\deg(D_0)$ is even}\\
q^{\deg(D_0)/2-1/2} &\text{if $\deg(D_0)$ is odd}
\end{cases}
\end{equation*} and 
\begin{equation*}
|\F_\f| q^g = \begin{cases}
q^{\deg(D)/2-1} &\text{if $\deg(D)$ is even}\\
q^{\deg(D)/2-1/2} &\text{if $\deg(D)$ is odd.}
\end{cases}
\end{equation*}
In summary, 
\begin{equation}\label{hstar}
h^*(a,\epsilon,\p) \geq \begin{cases}
\frac{1}{q}q^{\deg(D)/2} L(1,\xi) &\text{if $\deg(D)$ is even}\\
\frac{2}{\sqrt{q}}q^{\deg(D)/2}L(1,\xi) &\text{if $\deg(D)$ is odd}
\end{cases}
\end{equation}
\begin{lemma}\label{set_bound1} Let $\p \subset A$ be a prime ideal and $p$ its monic generator. For every $S \subseteq \{(a,\epsilon) \in A \times \F_q^\times | \deg(a^2-4\epsilon p) = \deg(\p)\}$,
\begin{equation*}
\sum_{(a,\epsilon) \in S}h^*(a,\epsilon,\p) \geq \begin{cases}
|S|\sqrt{|\F_{\p}|}(1-\frac{\deg(\p)}{\sqrt{q}})/q  &\text{if $\deg(\p)$ is even}\\
2|S|\sqrt{|\F_{\p}|}(1-\frac{\deg(\p)}{\sqrt{q}})/\sqrt{q} &\text{if $\deg(\p)$ is odd.}
\end{cases}
\end{equation*}
\end{lemma}
\begin{proof} We first lower bound $L(1,\xi)$ for a $\xi$ corresponding to an arbitrary $(a,\epsilon)$. 
$$L(1,\xi) = \frac{1}{(1-\xi(\infty)|\infty|^-1)} \prod_{\ell\ prime\ of\ A}\left(1- \xi(\ell) |\ell|^{-1} \right)^{-1}$$
$$ \geq  \prod_{\ell\ prime\ of\ A}\left(1- \xi(\ell) |\ell|^{-1} \right)^{-1}$$
since $\xi(\infty) \in \{0,1\}$ and $|\infty| = q$.
From Weil's proof of the Riemann hypothesis for curves over finite fields, there exists $w_i \in \mathbb{C}$ with $|w_i| = \sqrt{q}$ such that 
$$\prod_{\ell\ prime\ of\ A}\left(1- \xi(\ell) |\ell|^{-1} \right)^{-1} = \prod_{i=1}^{\deg(cond(\xi))-1}\left(1-\frac{w_i}{q}\right).$$
Since $\deg(cond(\xi)) -1  \leq \deg(\p)$, 
$$\prod_{i=1}^{\deg(cond(\xi))-1}\left(1-\frac{w_i}{q}\right) \geq \left(1 - \frac{1}{\sqrt{q}}\right)^{\deg(\p)} \geq 1- \frac{\deg{\p}}{\sqrt{q}}.$$
For all $(a,\epsilon) \in S$, since $\deg(a^2-4\epsilon p) = \deg(\p)$, inequality \ref{hstar} implies
\begin{equation*}
\sum_{(a,\epsilon) \in S}h^*(a,\epsilon,\p) \geq \begin{cases}
|S|\sqrt{|\F_{\p}|}(1-\frac{\deg(\p)}{\sqrt{q}})/q  &\text{if $\deg(\p)$ is even}\\
2|S|\sqrt{|\F_{\p}|}(1-\frac{\deg(\p)}{\sqrt{q}})/\sqrt{q} &\text{if $\deg(\p)$ is odd.}
\end{cases}
\end{equation*}
\end{proof}
\begin{lemma}\label{set_bound2} Let $\p \subset A$ be a prime ideal and $p$ its monic generator. For every $S \subseteq \{(a,\epsilon) \in A \times \F_q^\times | \deg(a^2-4\epsilon p) = \deg(\p)\}$, the number of Drinfeld modules $(\phi/\p)$ over $\F_p$ with $(a_{\phi,\p},\epsilon_{\phi,\p}) \in S$ is lower bounded by 
\begin{equation*}
\left|\{(\phi/\p)| (a_{\phi,\p},\epsilon_{\phi,\p}) \in S\}\right|  \geq \begin{cases}
|\F_\p^\times||S|\sqrt{|\F_{\p}|}(1-\frac{\deg(\p)}{\sqrt{q}})/(q(q-1))  &\text{if $\deg(\p)$ is even}\\
2|\F_\p^\times||S|\sqrt{|\F_{\p}|}(1-\frac{\deg(\p)}{\sqrt{q}})/(\sqrt{q}(q-1)) &\text{if $\deg(\p)$ is odd.}
\end{cases}
\end{equation*}
\end{lemma}
\begin{proof}
Fix an $S \subseteq \{(a,\epsilon) \in A \times \F_q^\times | \deg(a^2-4\epsilon p) = \deg(\p)\}$.
$$\left|\{(\phi/\p)| (a_{\phi,\p},\epsilon_{\phi,\p}) \in S\}\right|  = \sum_{(a,\epsilon) \in S} \sum_{(\phi/\p) \in H(a,\epsilon,\p)} \frac{|\F_p^\times|}{Aut_{\F_\p}(\phi)} $$
$$\ \ \ \ \ \ \ \ \ \ \ \ \ \ \ \ \ \ \ \ \ \ \  = \frac{|\F_\p^\times|}{q-1} \sum_{(a,\epsilon) \in S} h^*(a,\epsilon,\p).$$
Applying Lemma \ref{set_bound1}, we get
\begin{equation*}
\left|\{(\phi/\p)| (a_{\phi,\p},\epsilon_{\phi,\p}) \in S\}\right|  \geq \begin{cases}
|\F_\p^\times||S|\sqrt{|\F_{\p}|}(1-\frac{\deg(\p)}{\sqrt{q}})/(q(q-1))  &\text{if $\deg(\p)$ is even}\\
2|\F_\p^\times||S|\sqrt{|\F_{\p}|}(1-\frac{\deg(\p)}{\sqrt{q}})/(\sqrt{q}(q-1)) &\text{if $\deg(\p)$ is odd.}
\end{cases}
\end{equation*}
\end{proof}
\section{Degree Estimation and Euler-Poincare Characteristic of Drinfeld Modules}\label{degree_guessing_section}
\noindent Henceforth, let $h \in A$ denote the monic square free reducible polynomial whose factorization $h = \prod_i p_i$ into monic irreducible polynomials $p_i \in A$ we seek. By the chinese reminder theorem, $\F_\h = \prod_{i}\F_{\p_i}$ where $\h$ and $\p_i$s are the principal ideals generated by $h$ and the $p_i$s respectively. We next present a novel algorithm to compute the degree (call $s_h$) of the smallest degree factor of $h$ using Drinfeld modules. \\ \\ 
For a Drinfeld module $\phi$ that has reduction at each prime dividing $\h$, $$\phi(\F_\h) \cong \bigoplus_{i} \phi(\F_{\p_i}) \cong \bigoplus_{i} (\phi/\p_i)(\F_{\p_i}) \Rightarrow \chi_{\phi,\h} = \prod_{i} \chi_{\phi,\p_i} = \prod_i\left(p_i - (a_{\phi,\p_i} - 1)/\epsilon_{\phi,\p_i} \right).$$ 
Since $\forall i, \deg((a_{\phi,\p_i} - 1)/\epsilon_{\phi,\p_i}) \leq \deg(\p_i)/2$,\\ $$\chi_{\phi,\h} = h\ +\ terms\ of\ smaller\ degree.$$
In fact, 
$$h- \chi_{\phi,\h} =  \sum_{i: \deg(\p_i) = s_h} \left( \frac{a_{\phi,\p_i} - 1}{\epsilon_{\phi,\p_i}} \prod_{j \neq i}p_j\right) +\left( terms\ of\ degree < (\deg(\h) - \lceil s_h/2 \rceil)\right)$$
$$\Rightarrow \deg(h - \chi_{\phi,\h})  \leq \deg(h) - \lceil s_h/2 \rceil.$$
When $\phi$ is chosen at random, the equidistribution theorem of Gekeler suggests, with high probability,
$$\deg\left(\sum_{i: \deg(\p_i) = s_h} \left( \frac{a_{\phi,\p_i} - 1}{\epsilon_{\phi,\p_i}} \prod_{j \neq i}p_j\right)\right) =  \deg(h) - \lceil s_h/2 \rceil \Rightarrow \deg(h - \chi_{\phi,\h})  = \deg(h) - \lceil s_h/2 \rceil$$
leading to the following algorithm to compute $s_h$.
\begin{algorithm}\label{alg0}
\ \\ \textbf{Input :} Monic square free reducible polynomial $h \in A$ of degree $n$.
\begin{enumerate}
\item Choose a Drinfeld module $\phi$ by picking $g_\phi \in A$ and $\Delta_\phi \in A^\times$ each of degree less than $\deg(h)$ independently and uniformly at random. 
\item  If $\gcd(\Delta_\phi,h) \neq 1$, output it as a factor. Else $\phi$ has reduction at primes dividing $\h$ and we proceed.
\item Compute $\chi_{\phi,\h}$.
\item \textbf{Output:} $ n- \deg(h-\chi_{\phi,\h})$.
\end{enumerate}
\end{algorithm}
\noindent The running time of the algorithm is dominated by step $\textit{(3)}$. One way to compute $\chi_{\phi,\h}$ is as the characteristic polynomial of $\phi_t$ viewed as a linear transformation on $\F_\h$. Computing characteristic polynomials of linear transformations over finite fields can be performed in polynomial time \cite{sto}.\\ \\
The output  is at least $\lceil s_h/2 \rceil$. We prove in the ensuing lemma that when $q$ is large enough compared to $n$, the output is $\lceil s_h/2\rceil$ with probability at least $1/4$. We are thus ensured of finding $\lceil s_h/2 \rceil$ with probability $1-\delta$ with only $O(\log(1/\delta))$ repetitions of the algorithm. From $\lceil s_h/2 \rceil$, we infer that $s_h$ is either $2 \lceil s_h/2 \rceil - 1$ or $2 \lceil s_h/2 \rceil$. We can test which one is correct by checking if $\gcd(t^{q^{s_h}}-t,h)$ is non trivial and further extract the product of factors of degree $s_h$.
\begin{remark}\label{large_q}
In the analysis of our algorithms, we may assume without loss of generality that $q \geq c_1n^{c_2}$ for some absolute positive constants $c_1$ and $c_2$. If $q$ were smaller, we could choose the smallest prime $c$ that satisfies $q^c \geq c_1n^{c_2}$ and obtain the factorization in $\F_{q^c}[t]$ with the running time unchanged up to polylogarithmic factors in $n$.  Factors of $h$ irreducible over $\F_q$ and of degree prime to $c$ remain irreducible over $\F_{q^c}$. Factors of $h$ that are irreducible over $\F_q$ of degree (say $d$) divisible by $c$ will split into $c$ distinct irreducible factors over $\F_{q^c}$.  From the factorization over $\F_{q^c}$ thus obtained, express $h$ as $h = \prod_i h_i \prod_{d} h_d$ where $h_i \in \F_q[t]$ are irreducible factors that remained irreducible over $\F_{q^c}$ and $h_d \in \F_{q^c}[t]$ is the product of all irreducible factors of $h$ in $\F_{q^c}[t]$ (but not in $\F_q[t]$) of degree $d/c$.  In fact, $h_d$ is the product of all $\F_q[t]$ irreducible factors of $h$ of degree $d$ and hence $h_d \in \F_q[t]$. We may perform equal degree factorization on $h_d$ to obtain all $\F_q[t]$ irreducible degree $d$ factors of $h$. Since $c$ is bounded by an absolute constant, the post processing steps after obtaining the factorization over $\F_{q^c}$ take at most $O(n^{1+o(1)} (\log q)^{1+o(1)} + n (\log q)^{2+o(1)})$ expected time.
\end{remark}
\begin{lemma}\label{alg_0_proof} 
If $q$ is odd and $\sqrt{q} \geq 2 n$, algorithm \ref{alg0} outputs $\lceil s_h/2 \rceil$ with probability at least $1/4$.
\end{lemma}
\begin{proof}
For the output $\deg(h) - \deg(h - \chi_{\phi,\h})$ of algorithm \ref{alg0} to be $\lceil s_h/2 \rceil$, it suffices for $$\deg\left(\sum_{i: \deg(\p_i) = s_h} \left( \frac{a_{\phi,\p_i} - 1}{\epsilon_{\phi,\p_i}} \prod_{j \neq i}p_j\right)\right) =  \deg(h) - \lceil s_h/2 \rceil$$ to hold. Since $p_i$ are all monic, this is equivalent to \begin{equation}\label{sum_condition} \sum_{i: \deg(\p_i) = s_h}  \frac{a_{\phi,\p_i,\lfloor s_h/2 \rfloor}}{\epsilon_{\phi,\p_i}} \neq 0 \end{equation} where $a_{\phi,\p_i,\lfloor s_h/2 \rfloor} \in \F_q$ denotes the coefficient of the $t^{\lfloor s_h/2 \rfloor}$ term in $a_{\phi,\p_i}$.\\ \\
Fix a factor $\p_j$ of $\h$ of degree $s_h$.\\ \\
Since $g_\phi \in A$ and $\Delta_\phi \in A^\times$ are each chosen of degree less than $\deg(h)$ independently and uniformly at random (with $\gcd(\Delta_\phi,h) = 1$), by the chinese remainder theorem, the tuple $\{(g_\phi \mod \p_i, \Delta_\phi \mod \p_i)\}_i$ is distributed uniformly in $\prod_i(\F_{\p_i} \times \F_{\p_i}^\times)$.\\ \\
In particular, $\forall i\neq j$, $(a_{\phi,\p_i},\epsilon_{\phi,\p_i})$ and $(a_{\phi,\p_j},\epsilon_{\phi,\p_j})$ are independent and 
 $$Pr\left(\sum_{i: \deg(\p_i) = s_h}  \frac{a_{\phi,\p_i,\lfloor s_h/2 \rfloor} }{\epsilon_{\phi,\p_i}} \neq 0\right) = \sum_{\theta \in \F_q}\left(Pr\left(\sum_{i: \deg(\p_i) = s_h, i \neq j}  \frac{a_{\phi,\p_i,\lfloor s_h/2 \rfloor}}{\epsilon_{\phi,\p_i}}= -\theta \right)Pr\left(\frac{a_{\phi,\p_j,\lfloor s_h/2 \rfloor} }{\epsilon_{\phi,\p_j}} \neq \theta \right)\right)$$
\begin{equation}\label{min}\Rightarrow Pr\left(\sum_{i: \deg(\p_i) = s_h}  \frac{a_{\phi,\p_i,\lfloor s_h/2 \rfloor} }{\epsilon_{\phi,\p_i}} \neq 0\right)  \geq \min_{\theta \in \F_q} Pr\left(\frac{a_{\phi,\p_j,\lfloor s_h/2 \rfloor} }{\epsilon_{\phi,\p_j}} \neq \theta \right).
\end{equation}
Fix a $\theta \in \F_q$ and let $$S_{\theta}:=\{(a,\epsilon) \in A \times \F_q^\times | \deg(a^2-4\epsilon p_j) = \deg(p_j), a_{\lfloor s_h/2 \rfloor}/\epsilon \neq \theta\}$$ where $a_{\lfloor s_h/2 \rfloor} \in \F_q$ denotes the coefficient of the $t^{\lfloor s_h/2 \rfloor}$ term in $a$. \\ \\
We next count elements in $S_\theta$ to lower bound its size. The condition $\deg(a^2-4\epsilon p_j) = \deg(p_j)$ implies that we are only allowed to pick $a \in A$ of degree at most $\deg(p_j)/2$. If $\deg(p_j)$ is odd, $\deg(a^2-4\epsilon p_j) = \deg(p_j)$ is always satisfied for every $a \in A$ of degree at most $\deg(p_j)/2$. To pick an element in $S_\theta$, but for the coefficient $a_{\lfloor s_h/2 \rfloor}$, we may choose the coefficients of $a \in A$ arbitrarily of degree at most $\deg(\p_j)/2$ and arbitrarily choose $\epsilon \in \F_q^\times$. For each such choice, to satisfy $\deg(a^2-4\epsilon p_j) = \deg(p_j)$ and $a_{\lfloor s_h/2 \rfloor}/\epsilon \neq \theta$, we need to exclude at most two choices for $a_{\lfloor s_h/2 \rfloor}$ if $\deg(p_j)$ is even and at most one choice if $\deg(p_j)$ is odd. Thus
\begin{equation*}
 \left|S_\theta\right|  \geq \begin{cases}
(q-2) \sqrt{|\F_{\p_j}|} q   &\text{if $\deg(p_j)$ is even}\\
(q-1) \sqrt{|\F_{\p_j}|} \sqrt{q} &\text{if $\deg(p_j)$ is odd.}
\end{cases}
\end{equation*}
Applying Lemma \ref{set_bound2} for $S_\theta$, we get 
$$\left|\{(\phi/\p_j)| (a_{\phi,\p_j},\epsilon_{\phi,\p_j}) \in S_\theta \}\right|  \geq \left(1-\frac{\deg(p_j)}{\sqrt{q}}\right) \left(1-\frac{1}{q-1}\right) |\F_{\p_j}| |\F_{\p_j}^\times|.$$
Thus, for $\sqrt{q} \geq 2 \deg(h) \geq 2 \deg(p_j)$, $$Pr\left(\frac{a_{\phi,\p_j,\lfloor s_h/2 \rfloor}}{\epsilon_{\phi,\p_j}} \neq \theta \right) \geq \left(1-\frac{\deg(p_j)}{\sqrt{q}}\right) \left(1-\frac{1}{q-1}\right) \geq \frac{1}{4}$$ and by equation \ref{min} the lemma follows. 
\end{proof}
\noindent \textbf{Proof of Theorem \ref{chi_reduction_theorem} and Corollary \ref{chi_factor_corollary}:} Theorem \ref{chi_reduction_theorem} follows from the proof of correctness (Lemma \ref{alg_0_proof}) of Algorithm \ref{alg0}. Corollary \ref{chi_factor_corollary} follows by considering lines one and three in algorithm \ref{alg0} to be performed by the black box $\mathcal{B}$ in Corollary \ref{chi_factor_corollary}. The case $\gcd(\Delta_{\phi},h) \neq 1$  for a random $\phi$ is unlikely to happen and hence step $2$ of algorithm \ref{alg0} may be ignored in the reduction.\\ \\
When $q$ is large enough (say $q \geq 2 \deg(h)^4$), it is likely for a randomly chosen $\phi$ that $\phi(\F_\h)$ is a cyclic $A$-module. Further, for a random $\alpha \in \phi(\F_\h)$, it is likely that $Ord(\alpha) = Ann(\phi(\F_\h))$. Since $\phi(\F_\h)$ being cyclic implies $\chi_{\phi,\h} = Ann(\phi(\F_\h))$, it is likely that $Ord(\alpha)  = \chi_{\phi,\h}$. Thus, instead of computing $\chi_{\phi,\h}$ in Algorithm \ref{alg0}, we could compute $Ord(\alpha)$ for a random $\alpha \in \phi(\F_\h)$ and be assured that the output $\deg(h) - \deg(h - Ord(\alpha))$ likely is $\lceil s_h/2 \rceil$. 
\begin{algorithm}\label{alg1}
\ \\ \textbf{Input :} Monic square free reducible polynomial $h \in A$ of degree $n$.
\begin{enumerate}
\item Choose a Drinfeld module $\phi$ by picking $g_\phi \in A$ and $\Delta_\phi \in A^\times$ each of degree less than $\deg(h)$ independently and uniformly at random. 
\item  If $\gcd(\Delta_\phi,h) \neq 1$, output it as a factor. Else $\phi$ has reduction at primes dividing $\h$ and we proceed.
\item Choose $\alpha \in \phi(\F_\h)$ at random and compute $Ord(\alpha)$ with constant probability.
\item If $\deg(Ord(\alpha)) = \deg(h)$, \textbf{Output:} $ n- \deg(h-Ord(\alpha))$.
\end{enumerate}
\end{algorithm}
\noindent Every step except for $\textit{(3)}$ can be performed in $O(n \log q)$ time. In \S~\ref{order_computation}, we show that the order of an element in $\phi(\F_\h)$ (and hence step $\textit{(3)}$) can be computed with probability arbitrarily close to $1$ in $O(n^{(1+\omega)/2} +o(1) (\log q)^{1+o(1)} + n^{1+o(1)} (\log q)^{2+o(1)})$ expected time. In the subsequent lemma, we prove a lower bound on the probability that reductions of Drinfeld modules are cyclic and use it in Theorem \ref{alg_1_proof} to prove that Algorithm \ref{alg1} outputs $s_h$ with constant probability. Consequently, we have an $O(n^{(1+\omega)/2} +o(1) (\log q)^{1+o(1)} + n^{1+o(1)} (\log q)^{2+o(1)})$ expected time algorithm to extract a non trivial factor.
\begin{lemma}\label{cyclic_lemma}
For odd $q$, for every prime ideal $\p \subset A$, the probability that $\phi(\F_\p)$ is a cyclic $A$-module for a randomly chosen $\phi/\p$ is at least $\left(1-\frac{\deg(\p)+1/2}{2(q-1)}\right)$. 
\end{lemma}
\begin{proof} Let $\p \subset A$ be a prime ideal and $p$ its monic generator. Cojocaru and Papikian \cite[Cor 3]{cp} determined the following precise characterization of the $A$-module structure of finite Drinfeld modules when $\F_q$ is of odd characteristic. For a Drinfeld module $\phi$ with reduction at $\p$, let $f_{\phi,\p} \in A$ denote the largest monic square factor of the discriminant $a_{\phi,\p}^2-4 \epsilon_{\phi,\p}p$ of $P_{\phi,\p}$. As $A$-modules $$\phi(\F_\p) \cong  A/(m_{\phi,\p}) \oplus A/(m_{\phi,\p} n_{\phi,\p})$$ where $$m_{\phi,\p} = \gcd(f_{\phi,\p}, a_{\phi,\p}-2).$$ In particular, $\phi(\F_{\p})$ is $A$-cyclic if and only if $\gcd(f_{\phi,\p},a_{\phi,\p}-2)=1$. Let $$S_{\p}:=\{(a,\epsilon) \in A \times \F_q^\times | \deg(a^2-4\epsilon p) = \deg(p), \gcd(a^2-4\epsilon p, a-2)=1\}.$$ We next estimate the size of $S_{\p}$. An element in $S_{\p}$ can be chosen as follows. Pick $a \in A$ arbitrarily of degree at most $\deg(\p)/2$. For such a chosen $a$, to satisfy $\gcd(a^2-4\epsilon p, a-2)=1$, pick $\epsilon$ such that for all monic irreducible polynomials $\ell$ dividing $a-2$, $a^2-4\epsilon p \neq 0 \mod \ell$. For a fixed monic irreducible $\ell$ dividing $a-2$, there is at most one $\epsilon \in \F_q^\times$ such that $a^2-4\epsilon p = 0 \mod \ell$ for if there were two, then that would imply $\ell$ divides $p$ which contradicts the fact that $p$ is irreducible and of degree higher than $\ell$. Thus, for a chosen $a$, to ensure $a^2-4\epsilon p$ and $a-2$ are relatively prime, we need to exclude at most $\deg(\p)/2$ values for $\epsilon$. When $\deg(\p)$ is even, for a chosen $a$, to ensure $\deg(a^2-4\epsilon p) = \deg(\p)$, we need to exclude at most one choice for $\epsilon$.
\begin{equation*}
 \Rightarrow \left|S_{\p}\right|  \geq \begin{cases}
(q-2-\deg(\p)/2) \sqrt{|\F_{\p_j}|} q   &\text{if $\deg(\p)$ is even}\\
(q-1-\deg(\p)/2) \sqrt{|\F_{\p_j}|} \sqrt{q} &\text{if $\deg(\p)$ is odd.}
\end{cases}
\end{equation*}
Applying Lemma \ref{set_bound2} for $S_{\p}$, we get  
$$\left|\{(\phi/\p)| (a_{\phi,\p},\epsilon_{\phi,\p}) \in S_{\p} \}\right|  \geq \left(1-\frac{\deg(\p)+1/2}{2(q-1)}\right) |\F_{\p}| |\F_{\p}^\times|.$$
Thus $\phi(\F_{\p})$ is $A$-cyclic with probability at least $\left(1-\frac{\deg(\p)+1/2}{2(q-1)}\right).$
\end{proof}
\begin{theorem}\label{alg_1_proof} There exists a positive constant $c$ such that for $q$ odd and at least $2n^4$, algorithm \ref{alg1} outputs $\lceil s_h/2 \rceil$ with probability at least $c$.
\end{theorem}
\begin{proof}
Assume $q$ is odd and $q \geq 2n^4$. 
For a choice of $\alpha$ and $\phi$ made in algorithm \ref{alg1}, if the following three conditions hold, then clearly the output is $\lceil s_h/2\rceil$.
\begin{itemize}
\item $Ord(\alpha) = Ann(\phi(\F_\h))$,
\item $Ann(\phi(\F_\h)) = \chi_{\phi,\h}$,
\item $\deg(h) - \deg(h - \chi_{\phi,\h})= \lceil s_{h}/2 \rceil$.\\
\end{itemize}
For a fixed $\phi$ and a random $\alpha \in \phi(\F_\h)$, from the $A$-module decomposition of $\phi(\F_\h)$ into invariant factors, we infer that $Ord(\alpha) = Ann(\phi(\F_\h))$ with probability at least 
$$\frac{ \left|\{a \in A | \deg(a) < \deg(Ann(\phi(\F_\h))), \gcd(a, Ann(\phi(\F_\h))) =1\}\right|}{q^{\deg(Ann(\phi(\F_\h)))}} \geq (1-1/q)^{\deg(h)} \geq 1-\deg(h)/q \geq \frac{1}{2}.$$
The last inequality is a consequence of $q \geq 2\deg(h)^4$. If $\phi(\F_\h)$ is a cyclic $A$-module, then by definition $\chi_{\phi,\h} = Ann(\phi(\F_\h))$. Hence, to claim the theorem, it suffices to prove that for a Drinfeld module $\phi$ chosen at random as in algorithm \ref{alg1}, the following two conditions hold with constant probability 
\begin{itemize}
\item $\phi(\F_\h)$ is a cyclic $A$-module,
\item $\deg(h) - \deg(h - \chi_{\phi,\h})= \lceil s_{h}/2 \rceil$.\\
\end{itemize}
The proof proceeds by induction on the factors of $h$. Let $m$ denote the number of irreducible factors of $h$. Without loss of generality relabel the irreducible factors of $h$ such that  $h = \prod_{i=1}^m p_i$ where $\deg(p_1) \geq \deg(p_2) \geq \ldots \geq \deg(p_m)$. In particular, $\deg(p_m) = s_h$. Let $h_i := \prod_{j=1}^i p_j$ and $\h_i := (h_i)$. \\ \\
Induction Hypothesis: For $i<m$, assume $\phi(\F_{\h_i})$ is a cyclic $A$-module with probability at least $\left(1-\frac{\deg(h_i) }{\sqrt{q}}\right)^{i-1}$.\\ \\
The initial case $i=1$ of the induction hypothesis (that is,  $\h_1$ is prime), follows from Lemma \ref{cyclic_lemma}.\\ \\
We next lower bound the probability that $\phi(\F_\h)$ ($=\phi(\F_{\h_m})$) is $A$-cyclic and $\deg(h) - \deg(h - \chi_{\phi,\h})) = \lceil s_{h}/2 \rceil$ conditioned on $\phi(\F_{h_{m-1}})$ being $A$-cyclic. Since $g_\phi \in A$ and $\Delta_\phi \in A^\times$ are each chosen of degree less than $\deg(h)$ independently and uniformly at random (with $\gcd(\Delta_\phi,h) = 1$), by the chinese remainder theorem, the tuple $\{(g_\phi \mod \p_i, \Delta_\phi \mod \p_i)\}_i$ is distributed uniformly in $\prod_i(\F_{\p_i} \times \F_{\p_i}^\times)$. 
In particular, $(a_{\phi,\p_m},\epsilon_{\phi,\p_m})$ is independent of the structure of $\phi(\F_{\h_{m-1}})$. Consequently, instead of conditioning on $\phi(\F_{\h_{m-1}})$ being $A$-cyclic, we fix a tuple $$ \prod_{i <m} (g_\phi \mod \p_i, \Delta_\phi \mod \p_i)$$ such that $\phi(\F_{\h_{m-1}})$ is $A$-cyclic. In particular, $\chi_{\phi,\h_{m-1}}$ is fixed. In the remainder of the proof, the only randomness in $\phi$ comes from choosing $(g_\phi \mod \p_m, \Delta_\phi \mod \p_m)$ uniformly at random from $\F_{\p_m} \times \F_{\p_m}^\times$.\\ \\
As reasoned in the proof of Lemma \ref{alg_0_proof}, for $\deg(h) - \deg(h - \chi_{\phi,\h})$ to be $\lceil s_h/2 \rceil$, it suffices for  \begin{equation}\label{sum_condition} \sum_{i: \deg(\p_i) = s_h}  (a_{\phi,\p_i,\lfloor s_h/2 \rfloor} /\epsilon_{\phi,\p_i}) \neq 0 \end{equation} to hold where $a_{\phi,\p_i,\lfloor s_h/2 \rfloor} \in \F_q$ denotes the coefficient of the $t^{\lfloor s_h/2 \rfloor}$ term in $a_{\phi,\p_i}$.\\ \\
Since $\forall i \neq j, (g_\phi \mod \p_i, \Delta_\phi \mod \p_i)$ is fixed, $$\theta : = \sum_{i: \deg(\p_i) = s_h, i \neq m}  (a_{\phi,\p_i,\lfloor s_h/2 \rfloor} /\epsilon_{\phi,\p_i})$$
is fixed. Clearly equation \ref{sum_condition} holds if and only if $(a_{\phi,\p_m,\lfloor s_h/2 \rfloor} /\epsilon_{\phi,\p_m}) \neq \theta$. \\ \\
Since $\phi(\F_{\h_{m-1}})$ is $A$-cyclic, if $\phi(\F_{\p_m})$ is $A$-cyclic and $\chi_{\phi,\p_m}$ is relative prime to $\chi_{\phi,\h_{m-1}}$, then $\phi(\F_\h)$ is $A$-cyclic. Thus, to ensure $\phi(\F_\h)$ is $A$-cyclic, it suffices if $\gcd(p_m - (a_{\phi,\p_m}-1)\epsilon_{\phi,\p_m},\chi_{\phi,\h_{m-1}})=1$ and $\gcd(a_{\phi,\p_m}^2-4\epsilon_{\phi,\p_m}p_m, a_{\phi,\p_m}-2)=1$. The argument as to why is identical to the discussion in the proof of Lemma \ref{cyclic_lemma}.\\ \\
For $a \in A$, denote by $a_{\lfloor s_h/2 \rfloor} \in \F_q$ the coefficient of $t^{\lfloor s_h/2 \rfloor}$ in $a$. To summarize, the set $S_{p_m}$ of tuples $(a,\epsilon) \in A\times \F_q^\times$ that satisfy the four conditions
\begin{enumerate}[(i)]
\item $\deg(a^2-4\epsilon p_m) = \deg(p_m)$,
\item $a_{\lfloor s_h/2\rfloor} \neq \theta \epsilon$,
\item $\gcd(\epsilon p_m - (a-1),\chi_{\phi,\widehat{\h}}) = 1$,
\item $\gcd(a^2-4\epsilon p_m,a-2)=1$,
\end{enumerate} 
has the property that if $(a_{\phi,\p_m},\epsilon_{\phi,\p_m}) \in S_{\p_m}$ then $\phi(\F_\h)$ is $A$-cyclic and $\deg(h) - \deg(h - \chi_{\phi,\h}) = \lceil s_{h}/2 \rceil$.\\ \\
We estimate a lower bound on the size of $S_{\p_m}$ by choosing $a \in A$ arbitrarily of degree at most $\deg(p_m)/2$ and for each choice of $a$, picking only those $\epsilon$ such that the four conditions are satisfied.
To ensure the first condition, we need to exclude at most one choice each for $\epsilon$. To satisfy the fourth condition, for a fixed choice of $a$, we need to exclude at most $\deg(a-2) \leq \deg(p_m)/2$ choices for $\epsilon$. This is because for each monic irreducible polynomial $\ell$ dividing $a-2$, there is at most one $\epsilon \in \F_q^\times$ such that $a^2 - 4 \epsilon p_m = 0 \mod \ell$. For if there were two, then $\ell$ would divide $p_m$, which is an irreducible polynomial of degree higher than $\deg(\ell)$. To satisfy the third condition, for a fixed choice of $a \neq 1$, we need to exclude at most $\deg(\h_{m-1})$ choices for $\epsilon$. This is because for each monic irreducible polynomial $\ell$ dividing $\chi_{\phi,\h_{m-1}}$, there is at most one $\epsilon \in \F_q^\times$ such that $\epsilon p_m - (a-1) = 0 \mod \ell$. For if there were two, then $\ell$ would divide $p_m$ in which case restricting to $a \neq 1$ assures that $\epsilon p_m - (a-1) \neq 0 \mod \ell$. 
If $\theta \neq 0$, to satisfy the second condition we need to exclude at most one choice for $\epsilon$ which implies 
\begin{equation*}
\theta \neq 0 \Rightarrow \left|S_{\p_m}\right|  \geq \begin{cases}
\sqrt{|\F_{\p_m}|-1} (q-\deg(h_{m-1})-\deg(p_m)/2-3) q  &\text{if $\deg(p_m)$ is even}\\
\sqrt{|\F_{\p_m}|-1} (q-\deg(h_{m-1})-\deg(p_m)/2-3) \sqrt{q} &\text{if $\deg(p_m)$ is odd.}
\end{cases}
\end{equation*}
For $\theta=0$, the second condition can be satisfied by restricting the count to non zero $a_{\lfloor s_h/2 \rfloor}$ and we get 
\begin{equation*}
\theta = 0 \Rightarrow \left|S_{\p_m}\right|  \geq \begin{cases}
\sqrt{|\F_{\p_m}|-1} (q-\deg(h_{m-1})-\deg(p_m)/2-2) (q-1)/q  &\text{if $\deg(p_m)$ is even}\\
\sqrt{|\F_{\p_m}|-1} (q-\deg(h_{m-1})-\deg(p_m)/2-2) (q-1)/\sqrt{q} &\text{if $\deg(p_m)$ is odd.}
\end{cases}
\end{equation*}
Applying Lemma \ref{set_bound2} for $S_{\p_m}$ and assuming $q \geq 2\deg(h)^4$, we get 
$$\left|\{(\phi/\p_m)| (a_{\phi,\p_m},\epsilon_{\phi,\p_m}) \in S_{\p_m} \}\right|  \geq \left(1-\frac{\deg(h)}{\sqrt{q}}\right) |\F_{\p_m}| |\F_{\p_m}^\times|.$$
The probability that $\phi(\F_\h)$ is $A$-cyclic and $\deg(h) - \deg(h - Ann(\phi(\F_{\h}))) = \lceil s_{h}/2 \rceil$ conditioned on $\phi(\F_{\h_{m-1}})$ being $A$-cyclic is hence at least $$1-\frac{\deg(h)}{\sqrt{q}}.$$
By induction $\phi(\F_\h)$ is $A$-cyclic and $\deg(h) - \deg(h(t) - Ann(\phi(\F_{\h}))) = \lceil s_{h}/2 \rceil$ with probability at least $$\left(1-\frac{\deg(h) }{\sqrt{q}}\right)^{m} \geq 1 -\frac{m\deg(h)}{\sqrt{q}} \geq 1 -\frac{\deg(h)^2}{\sqrt{q}}$$ which is lower bounded by a constant since $q \geq 2\deg(h)^4$ and the theorem follows.
\end{proof}
\subsection{Order Finding in Finite Drinfeld Modules}\label{order_computation}
We sketch a Montecarlo randomized algorithm to compute $Ord(\alpha)$ with probability arbitrarily close to $1$ in $O(n^{(1+\omega)/2+o(1)}(\log q)^{1+o(1)} + n^{1+o(1)} (\log q)^2)$ time\footnote{We may replace $(1+\omega)/2$ with $\omega_2/2$, where $\omega_2$ is the exponent of $n \times n$ by $n \times n^2$ matrix multiplication (see \cite{ku}).}. The algorithm works for every $\phi$ with reduction at $\h$ and every $\alpha \in \phi(\F_h)$.  Fix a $\phi$ with reduction at $\h$ and an $\alpha \in \phi(\F_\h)$. Compute the minimal polynomial of the linear sequence $\{\mathcal{U}(\phi_t^j(\alpha)), j \in \Z_{\geq 0}\}$ where $\mathcal{U}:\F_\h \longrightarrow \F_q$ is a random $\F_q$ linear map. The minimal polynomial of the linear sequence divides $Ord(\alpha)$ and with probability at least half equals $Ord(\alpha)$. Hence the least common multiple of the minimal polynomials of the resulting linear sequences of $O(\log(\delta))$ independent trials is $Ord(\alpha)$ with probability at least $1- \delta$. For a trial, the minimal polynomial of a sequence can be computed in $O(n^{1+o(1)}\log q)$ time using the fast Berlekamp Massey algorithm given the first $2\deg(h)-1$ elements in the sequence. Hence the critical step is the computation of 
\begin{equation}\label{auto}
\{\mathcal{U}(\alpha),\mathcal{U}(\phi_{t}(\alpha)),\mathcal{U}(\phi_{t}^2(\alpha)), \ldots, \mathcal{U}(\phi_{t}^{{2\deg(h)-2}}(\alpha))\} 
\end{equation}
for a randomly chosen $\mathcal{U}$.
This is virtually identical to the automorphism projection problem of Kaltofen-Shoup. The difference being that the Frobenius endomorphism modulo $\h$ is replaced by the Drinfeld endomophism $\phi_t$ modulo $\h$. In Kaltofen-Shoup apart from being an $\F_q$ linear endomorphism of $\F_h$, the only property of the Frobenius exploited is that $\tau(t \mod \h)$ can be computed in $\widetilde{O}(n^{1+o(1)} (\log q)^2)$ time using the vonzur Gathen-Shoup algorithm. To adapt the automorphism projection algorithm of Kaltofen and Shoup \cite{ks}[\S~3.2] to apply in our setting, we merely have to demonstrate how to efficiently compute $\phi_t( t \mod \h)$ given $\h$, $g_{\phi} \mod \h$ and $\Delta_\phi \mod \h$. Since $$\phi_t(t \mod \h) = t ^2\mod \h+ \tau(t \mod \h) + \tau^2(t \mod \h),$$
we can compute $\phi_t( t \mod \h)$ in $O(n^{1+o(1)} (\log q)^2)$ time with three Frobenius powers and two additions modulo $h$ thereby making the Katofen-Shoup algorithm applicable to our setting.
\subsection{Obtaining the Complete Factorization from a Factor Finding Procedure}\label{complete_factorization}
In this subsection, we prove that given access to a blackbox $\mathcal{D}$ that takes as input a square free $\f \in A$ and outputs an irreducible factor, there is an $O(n^{4/3+o(1)}  (\log q)^{1+o(1)})$ expected time algorithm $\mathcal{F}$ to factor a polynomial of degree $n$ over $\F_q$ into its irreducible factors. Further, this algorithm makes at most $n^{1/3}$ calls to $\mathcal{D}$. Thereby, Corollary \ref{chi_reduction_corollary} would follow from Theorem \ref{chi_reduction_theorem}. \\ \\
Without loss of generality, assume that the input $h \in A$ to $\mathcal{F}$ is square free and of degree $n$. Obtaining the factorization of $h$ by extracting one irreducible factor at a time using $\mathcal{D}$ could in the worst case take $\Theta(\sqrt{n})$ calls to $\mathcal{D}$. A faster alternative is to use the Kaltofen-Shoup algorithm with fast modular composition to extract small degree factors of $h$ and then invoke $\mathcal{D}$ to extract the large degree factors one at a time. In particular, using \cite[Lem 8.4, Thm 8.5]{ku}, extract all the irreducible factors of $h$ of degree at most $ n^{2/3}$ in $O(n^{4/3+o(1)}  (\log q)^{1+o(1)})$ expected time. The remaining irreducible factors of $h$ each have degree at least $\lceil n^{2/3} \rceil$. Hence there are at most $n^{1/3}$ irreducible factors of $h$ remaining and the complete factorization of $h$ can be obtained by extracting a factor at a time with at most $n^{1/3}$ calls to $\mathcal{D}$. 
\begin{remark}
Kaltofen and Shoup \cite[\S~3.1]{ks} through the blackbox Berlekamp algorithm \cite{kl} reduced polynomial factorization in time nearly linear in degree to two problems that are transposes of each other, namely automorphism projection and automorphism evaluation. Being transposes, a straight line program that computes $\F_q$ linear forms in the input for one would in linear time yield a straight line program for the other of the same complexity. In particular, there is nearly linear polynomial factorization algorithm if there is a nearly linear time $\F_q$ linear solution to the automorphism projection problem. Our order finding problem is no harder than automorphism projection. We hence arrive at the stronger assertion that polynomial factorization is reducible to automorphism projection. In particular, no assumptions on the $\F_q$ linearity of automorphism projection algorithm is made. We must however remark that the automorphism projection we consider (see equation \ref{auto}) is broader than that stated in \cite[\S~3.2]{ks} where only the Frobenius automorphism is considered.
\end{remark}
\subsection{Degree Estimation Using Carlitz Modules}
The degree estimation algorithm framework also gives rise to variants where there is no randomization with respect to the choice of Drinfeld modules. In fact, the following deterministic example from the author's Ph.D thesis \cite{nar} using Carlitz modules (rank $1$ Drinfeld modules) partly motivated the randomized version. The Carlitz module based algorithm is suited to the case when the characteristic of $\F_q$ does not divide the number of factors of the smallest degree. 
\begin{example} \textbf{Factor Degree Estimation using Carlitz Modules.}
\ \\ \textbf{Input :} Monic square free reducible polynomial $h \in A$.
\begin{enumerate}
\item Choose the Carlitz module $\phi$ (the rank $1$ Drinfeld module $\phi$ with $g_\phi = 1$ and $\Delta_\phi=  0$). 
\item Compute $\chi_{\phi,\h}$.
\item \textbf{Output:} $\deg(h) - \deg(h - \chi_{\phi,\h})$.
\end{enumerate}
In \cite{nar} it is shown that the output is exactly $s_h$ provided the number of factors of degree $s_h$ is not divisible by the characteristic of $\F_q$. This is true since for the Carlitz module $\phi$, for all prime ideals $\p \subset A$, $\chi_{\phi,\p} = p-1$ where $p$ is the monic generator of $\p$.\\ \\ 
\noindent Curiously for the Carlitz module $\phi$, finding $\chi_{\phi,h}$ is easily seen to be no harder than factoring $h$. Computing $\chi_{\phi,h}$ is linear time reducible to factoring $h$ since given the factorization of $h$, it is trivial to write down $\chi_{\phi,h}$ in $O(\deg(h) \log q)$ time. 
\end{example}
\section{Factorization Patterns of Polynomials in Small Intervals}\label{factor_distribution}
\noindent Our analysis of the Drinfeld module analogue of the black-box Berlekamp algorithm relies on the degree distribution in factorization patterns of polynomials in short intervals which we study in this section.\\ \\
For a partition $\lambda$ of a positive integer $e$, let $C_\lambda:=\{\sigma \in S_e | \lambda_\sigma = \lambda \}$ denote its conjugacy class where $S_e$ is the symmetric group on $e$ elements and $\lambda_\sigma$ is the partition of $e$ induced by the factorization of $\sigma$ into disjoint cycles. Let $P(\lambda):= |C_\lambda|/|S_e|$. For a polynomial $f \in A$, let $\lambda(f)$ denote the partition of $\deg(f)$ induced by the degrees of the irreducible factors in the factorization of $f$ in $A$.\\ \\
For $f \in A$ and a positive integer $m < \deg(f)$, define the interval around $f$ corresponding to the degree bound $m$ as $$ \mathcal{I}_m(f) : = \{f + a | a \in A , \deg(a) \leq m\}.$$ For $\mathcal{I} \subset A$ where each polynomial in $\mathcal{I}$ is of degree exactly $d>1$ and a partition $\lambda$ of $d$, define $$B_{q}(\mathcal{I},\lambda) : = \left\{a \in \mathcal{I} | \lambda(a) = \lambda\right\}\ \ \ and\ \ \ \pi_{q}(\mathcal{I},\lambda) : = |B_{q}(\mathcal{I},\lambda)|.$$
Bank, Bary-Soroker and Rosenzweig \cite{bbr} recently proved the following theorem when the field size $q$ tends to infinity while $d$ is fixed.
\begin{theorem}\label{conj}(\cite[Thm 1]{bbr}) For all monic $f \in A$ of fixed degree $d$, for all positive integers $2<m<d$ and for all partitions $\lambda$ of $d$,
$$\pi_{q}(\mathcal{I}_m(f),\lambda) \sim P(\lambda) |\mathcal{I}_m(f)|\ \ \ \ as\ \ \ \  q \rightarrow \infty.$$
\end{theorem}
\noindent It is widely conjectured (see \cite{bbr}) that 
\begin{conjecture}\label{smoothness_conjecture}
For all monic $f \in A$ of degree $d$ such that $3<d<\sqrt{q}/2$ and for all partitions $\lambda$ of $d$,
$$\pi_{q}(\mathcal{I}_m(f),\lambda) \sim P(\lambda) |\mathcal{I}_m(f)|\ \ \ \ as\ \ \ \  q^d \rightarrow \infty.$$
\end{conjecture}
\noindent  In the next subsection \S~\ref{density_subsection}, by applying an effective Lang-Weil bound to the argument in \cite{bbr}, we prove an effective version of  Theorem \ref{conj} that holds for $\log q >  5 d \log(d)$.

\subsection{A High Dimensional Variant of the Function Field Chebotarev Density Theorem}\label{density_subsection}
\noindent Let $E$ denote the rational function field $\F_q(t_1,\ldots,t_m)$ in the $m$ indeterminates $t_1,\ldots,t_m$. Let $F/E$ be a finite Galois extension of $E$. Fix an algebraic closure $\overline{\F}_q$ of $\F_q$ and let $$\alpha:Gal(F/E) \longrightarrow Gal((\overline{\F}_q \cap F)/\F_q)$$ denote the restriction map. 
Let $V = Spec(\F_q[t_1,\ldots,t_m])$ and let $V_{ur}(\F_q) \subset V(\F_q)$ denote the subset of $\F_q$ rational places in $V$ that are etale in the extension $F/E$. Let $O_F$ denote the integral closure of $\F_q[t_1,\ldots,t_m]$ in $F$ and let $W = Spec(O_F)$. 
For a place $\B \in W$ lying above a place $\p \in V$ that is etale in $F/E$, let $\sigma_\B \in Gal(F/E)$ denote its Artin symbol. For a place $\p \in V$ that is etale in $F/E$, let $$\Theta_\p:=\{\sigma_\B | \B \in W, \B|\p\} \subseteq ker(\alpha)$$ denote the conjugacy class of Artin symbols above $\p$. 
\begin{lemma}\label{density}
If $q\geq 2(m+1)[F:E]^2$, for every conjugacy class $\Theta \subseteq ker(\alpha)$, $$ \left| | \{ \p \in P_{ur}(\F_q) | \Theta_\p = \Theta \}| -\frac{|\Theta|}{| ker(\alpha)|}q^m \right| \leq \frac{|\Theta|}{|ker(\alpha)|}\left(([F:E]-1)([F:E]-2)\frac{q^m}{\sqrt{q}} + 5 [F:E]^{13/3}q^{m-1}\right).$$
\end{lemma}
\begin{proof}
Fix a conjugacy class $\Theta \subseteq ker(\alpha)$ and let $U := \{ \p \in P_{ur}(\F_q) | \Theta_\p = \Theta \}$.\\ \\
Let $\rho: W \longrightarrow V$ denote the norm map from $W$ down to $V$. Applying \cite{bs}[Prop. 2.2] to (V,W,$\rho,\Theta)$ implies the existence of a smooth irreducible affine $\F_q$-variety $\widehat{W}$ and a finite separable morphism $\pi:\widehat{W} \longrightarrow V$ such that
\begin{enumerate}[(i)]
\item $\pi(\widehat{W}) = U$,
\item $\deg(\pi) = |ker(\alpha)|$,
\item $\forall \p \in U, |\pi^{-1}(\p) \cap \widehat{W}(\F_q)| = |ker(\alpha)|/|\Theta|$.
\end{enumerate}
Since $\pi:\widehat{W} \longrightarrow V$ is finite, $\widehat{W}$ and $V$ have the same dimension, namely $m$. Further, $W$ and $\widehat{W}$ are twists of each other \cite{bs}. As a consequence, $W$ and $\widehat{W}$ have the same degree, namely $[F:E]$.\\ \\
Bounding the size of $\widehat{W}(\F_q)$ using an effective Lang-Weil bound \cite{cm}, 
$$\left| |\widehat{W}(\F_q)| - q^m \right| \leq ([F:E]-1)([F:E]-2)\frac{q^m}{\sqrt{q}} + 5 [F:E]^{13/3}q^{m-1}.$$
Since $\pi(\widehat{W}) = U$ and $\forall \p \in U, |\pi^{-1}(\p) \cap \widehat{W}(\F_q)| = |ker(\alpha)|/|\Theta|$, $|\widehat{W}(\F_q)| = \frac{|ker(\alpha)||U|}{|\Theta|}$
$$\Rightarrow  \left| |U| - \frac{|\Theta|}{|ker(\alpha)|}q^m \right| \leq \frac{|\Theta|}{|ker(\alpha)|}\left(([F:E]-1)([F:E]-2)\frac{q^m}{\sqrt{q}} + 5 [F:E]^{13/3}q^{m-1}\right).$$
\end{proof}

\begin{theorem}\label{small_interval_bound} For every positive integer $m \geq 2$, for every monic $f \in A$ of degree greater than $m$ and for every partition $\lambda$ of $\deg(f)$, if $\log q \geq 5 \deg(f) \log(\deg(f))$ then
$$\left|  \pi_{q}(\mathcal{I}_{m}(f),\lambda) - P(\lambda) |\mathcal{I}_{m}(f)| \right| \leq \frac{1}{2}  P(\lambda) |\mathcal{I}_{m}(f)|.$$
\end{theorem}
\begin{proof}
Fix a monic non constant polynomial $f \in A$ of degree at least $m$ and let $$\ef_f := f(t) + \sum_{i=1}^{m} x_i t^{i-1}.$$ Since the indeterminate $x_1$ only appears in $\ef_f$ as the linear term $x_1$, $\ef_f$ is absolutely irreducible and separable in $t$. Thus the splitting field $F_f$ of $\ef_f$ over $E = \F_q(x_1,\ldots,x_m)$ is Galois. We will shortly apply Lemma \ref{density} to the extension $F_f/E$. Before doing so, we argue that $F_f/E$ is a geometric extension.\\ \\
The splitting field of $\ef_f$ over $\overline{\F}_q(x_1,\ldots,x_m)$ is the composite $F_f.\overline{\F}_q$ and we have
$$Gal(F_f.\overline{\F}_q/\overline{\F}_q(x_1,\ldots,x_m)) = Gal(F_f.\overline{\F}_q/\overline{\F}_q(x_1,\ldots,x_m)) \leq Gal(F_f/E) \leq S_{\deg(f)}.$$
By \cite{bbr}[Prop 3.6], $Gal(F_f.\overline{\F}_q/\overline{\F}_q(x_1,\ldots,x_m)) \cong S_{\deg(f)}$
$$\Rightarrow Gal(F_f.\overline{\F}_q/\overline{\F}_q(x_1,\ldots,x_m)) \cong Gal(F_f/E) \cong S_{\deg(f)} \Rightarrow F_f \cap \overline{\F}_q = \F_q.$$
Hence $F_f/E$ is a geometric extension. Since $Gal(F_f \cap \overline{\F}_q/\F_q)$ is trivial, the restriction map $$\alpha_f : Gal(F_f/E) \longrightarrow Gal(F_f \cap \overline{\F}_q/\F_q)$$ has kernel $$ker(\alpha_f)  = Gal(F_f/E)\cong S_{\deg(f)}.$$
Since $Gal(\overline{\F}_q/\F_q) = \langle \tau \rangle \cong \Z$ where $\tau$ is the $q^{th}$ power Frobenius, homomorphisms from $Gal(\overline{\F}_q/\F_q)$ to $S_{\deg(f)}$ are parametrized by the permutations in $S_{\deg(f)}$ they map $\tau$ to. That is, $\sigma \in S_{\deg(f)}$ corresponds to $\theta_\sigma \in Hom(Gal(\overline{\F}_q/\F_q), S_{\deg(f)})$ that takes $\tau$ to $\sigma$. \\ \\
Fix a partition $\lambda$ of $\deg(f)$. For the conjugacy class $\Theta_{\lambda} : = \{\theta_\sigma| \sigma \in C_\lambda\} \subseteq ker(\alpha_f)$, 
\begin{equation}\label{eqn1} \frac{|\Theta_\lambda|}{|\ker(\alpha)|}  = \frac{|C_\lambda|}{|S_{\deg(f)}|} = P(\lambda)\end{equation} where the first equality follows from the fact that $ker(\alpha)\cong S_{\deg(f)}$ and $|\Theta_\lambda|= |C_\lambda|$ and the second equality follows from the definition of $P(\lambda)$.\\ \\
Equation \ref{eqn1} together with Lemma \ref{density} applied to the extension $F_f/E$ yields 
$$ \left| | \{ \p \in V_{ur,f}(\F_q) | \Theta_\p = \Theta_\lambda \}| -P(\lambda)q^m \right| \leq 2P(\lambda)[F_f:E]^{m+1}q^{m/2}$$
where $V_{ur,f}(\F_q) \subseteq V(\F_q)$ is the set of $\F_q$-rational places in $E$ that are etale in $F_f$. 
Identifying $V(\F_q)$ with $\A^m(\F_q)$, a prime $\p=(a_1,\ldots,a_m) \in V_{ur,f}(\F_q)$ has $\Theta_\p = \Theta_\lambda$ if and only if $\lambda(\ef_f(a_1,\ldots,a_m,t)) \in C_\lambda$ (\cite[Lem 2.1]{bs}). Thus 
$$ \left| | \{ (a_1,\ldots,a_m) \in V_{ur,f}(\F_q)  |  \lambda(\ef_f(a_1,\ldots,a_m,x)) = \lambda \}| -P(\lambda)q^m \right| $$ $$\leq P(\lambda)\left((\deg(f)-1)(\deg(f)-2)\frac{q^m}{\sqrt{q}} + 5 \deg(f)^{13/3}q^{m-1}\right).$$
For $q>20\deg(f)^5$, $$\frac{(\deg(f)-1)(\deg(f)-2)}{\sqrt{q}} + \frac{5 \deg(f)^{13/3}}{q} \leq \frac{1}{2},$$
$$ \Rightarrow \frac{P(\lambda)}{2} |\mathcal{I}_{m}(f)| \leq  \pi_{q}(\mathcal{I}_{m}(f),\lambda) \leq  \frac{3P(\lambda)}{2} |\mathcal{I}_{m}(f)|  + \pi_q^{ra}(\mathcal{I}_{m}(f),\lambda)$$
where $$\pi_q^{ra}(\mathcal{I}_{m}(f),\lambda) : = \ \{ (a_1,\ldots,a_m) \in V(\F_q) \setminus V_{ur,f}(\F_q)  |  \lambda(\ef_f(a_1,\ldots,a_m,x)) = \lambda \}|$$
accounts for ramified primes with the factorization pattern corresponding to $\lambda$. The ramified part $\pi_q^{ra}(\mathcal{I}_{m}(f),\lambda)$ is bounded by the number of $\F_q$ points in the variety defined by the discriminant $\Delta_t(\ef_f) \in \F_q[x_1,\ldots,x_m]$ of $\ef_f$ with respect to $t$.  By expressing $\Delta_t(\ef_f)$ as the resultant of $\ef_f$ and its derivative with respect to $t$, we see, $\deg(\Delta_t(\ef_f)) \leq 2[F_f:E]-1$. Applying an effective version of the Lang-Weil bound \cite{cm}, $\pi_q^{ra}(\mathcal{I}_{m}(f),\lambda)$ turns out to be negligible in our computation.\\ \\
Since $[F_f:E] = \deg(f)!$, for $\log q > 5 \deg(f) \log(\deg(f))$,
$$ \frac{1}{2}P(\lambda) |\mathcal{I}_{m}(f)|\leq  \pi_{q}(\mathcal{I}_{m}(f),\lambda) \leq  \frac{3}{2}P(\lambda) |\mathcal{I}_{m}(f)|$$
and the theorem follows.
\end{proof}

\begin{remark} Recent break through algorithms for discrete logarithm computation \cite{jou}\cite{bgjt} over a small characteristic finite field (say $\F_{r^d}$) have the following initial polynomial search step. Given $r$ and $d$, search for $h_0,h_1 \in \F_{r^2}[t]$, each of degree $2$ such that the factorization of $h_1t^r-h_0$ over $\F_{r^2}[t]$ has an irreducible factor of degree $d$. The search is known to succeed only under heuristic assumptions. If Theorem \ref{small_interval_bound} were true for $q  \geq (n-1)^2$, then as a corollary (by setting $q=r^2$, $f=t^{r+1}$ and $m=2$), the search provably succeeds (even when $h_1$ is fixed as $h_1 = t$) without making any heuristic assumptions. More generally, if Theorem \ref{small_interval_bound} holds for $q  \geq c_1n^{c_2}$ for some positive absolute constants $c_1,c_2$, then the heuristic assumptions in the polynomial selection step (with appropriate modifications) may be removed.
\end{remark}
\section{Drinfeld Module Analog of Berlekamp's Algorithm}\label{blackbox_section}
\noindent We motivate the Drinfeld module analog of Berlekamp's algorithm with a brief description of Lenstra's algorithm for integer factorization. Pollard's \textit{p-1} algorithm \cite{pol} is designed to factor an integer that has a prime factor modulo which the multiplicative group has smooth order. Say for instance that a positive integer $n$ has a prime factor $p$ such that every prime power factor of $p-1$ is bounded by $b$. The algorithm proceeds by choosing a positive integer $B$ as the smoothness bound and computes $m$, the product of all prime powers bounded by $B$. A positive integer $a < n$ is then chosen at random. Assume $a$ is prime to $n$ for otherwise $\gcd(a,n)$ is a non trivial factor of $n$. If $B \geq b$, since $p-1$ divides $m$, $$a^m-1 = (a^{p-1})^{m/(p-1)} - 1 \cong 0 \mod p \Rightarrow p \mid a^m-1$$ and $\gcd(a^{m}-1,n)$ is likely a non trivial factor of $n$.\\ \\
The running time is at least exponential in the size of $B$. For typical $n$, $B$ needs to be as big as the smallest factor of $n$ and thus the running time is typically exponential in the size of the smallest factor of $n$.\\ \\ 
Lenstra's elliptic curve factorization algorithm \cite{len} factors every $n$ in (heuristic) expected time sub-exponential in the size of the smallest factor $p$ of $n$. A key insight of Lenstra was to substitute the multiplicative group $(\Z/p\Z)^\times$ in Pollard's \textit{p-1} algorithm with the group $E(\F_p)$ of $\F_p$ rational points of a random elliptic curve $E$ over $\F_p$. The running time depends on the smoothness of the group order $|E(\F_p)|$ for a randomly chosen $E$. The Hasse-Weil bound guarantees that $||E(\F_p)| - (p+1) | \leq 2\sqrt{p}$ and Lenstra proved that his algorithm runs in expected time subexponential in the size of $p$ assuming a heuristic on the probability that a random integer in the interval $[p+1-2\sqrt{p}, p+1 + 2\sqrt{p}]$ is smooth.\\ \\
Our algorithm can be thought of as an analogue of Berlekamp's algorithm wherein the Frobenius action is replaced with a random rank-2 Drinfeld action; much like Lenstra's algorithm is an analogue of Pollard's $p-1$ obtained through replacing the multiplicative group modulo a prime with a random elliptic curve group. Before outlining the algorithm, a few remarks regarding notation are in order. For a positive integer $b$, we call a polynomial $b$-smooth if all its irreducible factors are of degree at most $b$. For a Drinfeld module $\phi$ (with reduction at primes dividing $\h$) and $\beta \in \phi(\F_\h)$, by $\gcd(\beta,h)$ we really mean the gcd of $h$ and a lift of $\beta$ to $A$.
\begin{algorithmoutline1}\label{algout1}
\ \\ \textbf{Input :} Monic square free reducible polynomial $h \in A$.
\begin{enumerate}
\item Pick a smoothness bound $b \geq 1$. 
\item Choose a Drinfeld module $\phi$ at random by picking $g_\phi \in A$ and $\D_\phi \in A^\times$ each of degree less than $\deg(h)$ independently and uniformly at random.
\item Choose a random non zero $\alpha \in \phi(\F_\h)$ and compute $Ord(\alpha)$.
\item Find a monic $b$-smooth factor $f$ of $Ord(\alpha)$ (if one exists).
\item \textbf{Output:} $\gcd(\phi_{Ord(\alpha)/f}(\alpha),h)$ is likely a non trivial factor of $h$.
\end{enumerate}
\end{algorithmoutline1}
\noindent There is flexibility on how the smooth factor $f$ is determined once $b$ is chosen. One extreme is to set $f$ to be the largest $b$-smooth factor of $Ord(\alpha)$. The other, is to further factor the largest $b$-smooth factor of $Ord(\alpha)$ (recursively or by other means) and to set $f$ to one of the $b$-smooth irreducible factors of $Ord(\alpha)$. A rigorous analysis of the former choice with $b=1$ is in \S~\ref{blackbox_linear}. An informal discussion of why the algorithm is likely to succeed with the latter choice follows keeping in mind that $$\phi(\F_\h) \cong \bigoplus_{i} \phi(\F_{\p_i})\ \  , \ \  \chi_{\phi,\p_i} = p_i - (a_{\phi,\p_i}-1)/\epsilon_{\phi,\p_i},\forall i.$$ 
For $d>0$, a random polynomial of degree $d$ has a linear factor with at least constant probability. Assume the plausible hypothesis (which is  for large $q$ true by Theorem \ref{small_interval_bound_intro}) that for every $i$, the probability of a polynomial in the interval $I_{p_i} := \{p_i + a\ |\ a \in A,  \deg(a) \leq \deg(\p_i)/2\}$ around $p_i$ possessing a $b$-smooth factor roughly equals the probability of a random polynomial of degree $\deg(\p_i)$ possessing a $b$-smooth factor. This smoothness hypothesis along with the equidistribution of characteristic polynomials (equation \ref{hstar}) suggests it is likely for every $b > 0$ that there exists a $j$ such that $\chi_{\phi,\p_j} = p_j - (a_{\phi,\p_j} - 1)/\epsilon_{\phi,\p_j}$ has a $b$-smooth factor. Since $Ann(\phi(\F_{\p_j}))$ divides $\chi(\phi(\F_{\p_j}))$, it is thus likely that 
$Ann(\phi(\F_\h))$ (which is the least common multiple of $\{Ann(\phi(\F_{\p_i}))\}_i$) has a $b$-smooth factor. Assuming that is the case, since $\alpha$ is chosen at random, with probability at least $1-1/q$ there is a monic $b$-smooth polynomial dividing $Ord(\alpha)$. The algorithm picks one such monic irreducible factor $f$ of $Ord(\alpha)$. The fact that the reduction of $\phi$ at $\h$ is random and equidistribution of characteristic polynomials (equation \ref{hstar}) imply the likely existence of $k$ such that $f(t)$ does not divide $\chi_{\phi,\p_k}$. Consequently  $$\phi_{Ord(\alpha)/f}(\alpha) \cong 0 \mod \p_k\ \ ,\ \   \phi_{Ord(\alpha)/f}(\alpha) \ncong 0\mod \prod_{i \neq k}\p_i$$ and thus $\gcd(\phi_{Ord(\alpha)/f}(\alpha),h)$ is likely a non trivial factor of $h$.\\ \\ 
The computation of $Ord(\alpha)$ can be performed efficiently through linear algebra as discussed in \S~\ref{order_computation}. This is in stark contrast to the integer analog, where finding the order of an element in the multiplicative group modulo a composite appears hard. A consequence is that unlike Lenstra's algorithm, our success probability is reliant not on $\chi_{\phi,\p}$ being smooth but merely on it possessing a smooth factor. The running times are thus bounded by a polynomial in the problem size.
\subsection{Drinfeld Module Analogue of Berlekamp's Algorithm With Linear Smoothness}\label{blackbox_linear}
\noindent In this section we formally state and analyze the version of the Drinfeld analog of blackbox Berlekamp algorithm where the smooth factor chosen is the product of all linear factors of the order of a randomly chosen element in $\phi(\F_\h)$.
\begin{algorithm}\label{alg2}
\ \\ \textbf{Input :} Monic square free reducible polynomial $h \in A$ of degree $n$.
\begin{enumerate}
\item Choose a Drinfeld module $\phi$ by picking $g_\phi \in A$ and $\Delta_\phi \in A^\times$ each of degree less than $\deg(h)$ independently and uniformly at random. 
\item  If $\gcd(\Delta_\phi,h) \neq 1$, output it as a factor. Else $\phi$ has reduction at primes dividing $\h$ and we proceed.
\item Choose $\alpha \in \phi(\F_\h)$ at random and compute $Ord(\alpha)$.
\item Compute $f = \gcd(t^q-t,Ord(\alpha)).$
\item Compute $\phi_{Ord(\alpha)/f}(\alpha)$.
\item \textbf{Output:} $\gcd(h, \phi_{Ord(\alpha)/f}(\alpha))$.
\end{enumerate}
\end{algorithm}
\noindent The running time of the algorithm is dominated by steps $\textit{(3)}$ and $\textit{(5)}$. As in \S~\ref{order_computation} , step $\textit{(3)}$ can be performed with $O(n^{(1+w)/2+o(1)} (\log q)^{o(1)}+n^{1+o(1)}(\log q)^2)$ expected time by adapting the automorphism projection algorithm of Kaltofen-Shoup. Step $\textit{(5)}$ poses the transpose problem of step $\textit{(3)}$ and can be performed in identical expected time as step $3$ by the transposition principle (see \cite[\S~3.2]{ks}).\\ \\
The rest of the section is devoted to showing that algorithm \ref{alg2} outputs a non trivial factor of $h$ with constant probability. In fact, we prove something stronger in Lemma \ref{alg2_proof} by showing that there exists positive constants $c_1$ and $c_2$ such that  for every factor $\p_i$ of $\h$, $\p_i$ divides $\gcd(h, \phi_{Ord(\alpha)/f}(\alpha))$ with probability at least $c_1$ and $\p_i$ does not divide $\gcd(h, \phi_{Ord(\alpha)/f}(\alpha))$ with probability at least $c_2$. As a consequence, not merely a factor but the complete factorization of $h$ can be obtained by recursing the algorithm with recursion depth bounded by $O((\log n)^2)$ (see \cite[\S~3]{ks}).
\begin{lemma}\label{alg2_proof} There exists positive a constant $c$ such that, for $q$ odd and $\log q \geq 5 n \log n$, at the termination of algorithm \ref{alg2}, for every prime factor $\p$ of $\h$ with monic generator $p$, $p$ divides $\gcd(h, \phi_{Ord(\alpha)/f}(\alpha))$ with probability at least $c$ and $p$ does not divide $\gcd(h, \phi_{Ord(\alpha)/f}(\alpha))$ with probability at least $c$.
\end{lemma}
\begin{proof}
Fix a prime factor $\p$ of $\h$ with monic generator $p$. Assume $q$ is odd, $\log q \geq 3n \log n$ 
and let 
$$S_{in}:= \{(a,\epsilon) \in A\times \F_q^\times | \deg(a^2-4\epsilon p)=\deg(p), \gcd(t^q-t, p-(a-1)/\epsilon)=1\},$$  
$$S_{out}:= \{(a,\epsilon) \in A\times \F_q^\times | \deg(a^2-4\epsilon p)=\deg(p), \gcd(t^q-t, p-(a-1)/\epsilon) \neq 1\}.$$
Let $\Lambda$ denote the set of partitions of $\deg(p)$ that contain $1$ and let $\widehat{\Lambda}$ denote the set of partitions of $\deg(p)$ that do not contain $1$.\\ \\
Since $\deg(a^2-4\epsilon p)=\deg(p)$ is always true when $\deg(p)$ is odd and $\deg(a)\leq 2 \deg(p)$, by Theorem \ref{small_interval_bound} it follows for $\deg(p)$ odd that 
$$|S_{in}| \geq \left(\sum_{\lambda \in \Lambda}P(\lambda)\right) \frac{\sqrt{q}(q-1)}{2}\sqrt{|\F_\p|},$$
$$|S_{out}| \geq \left(\sum_{\lambda \in \widehat{\Lambda}}P(\lambda)\right)  \frac{\sqrt{q}(q-1)}{2}\sqrt{|\F_\p|}.$$
When $\deg(p)$ is even, since the characteristic of $\F_q$ is assumed odd, we can enforce $\deg(a^2-4\epsilon p)=\deg(p)$ by restricting the choice of $\epsilon$ to ensure $4\epsilon$ is not a square in $\F_q^\times$ and picking $a\in A$ arbitrarily of degree at most $\deg(p)/2$. There are at least $(q-1)/2$ such choices for $\epsilon$ and applying Theorem \ref{small_interval_bound} once for each such choice we get for $\deg(p)$ even and $\F_q$ of odd characteristic, 
$$2|S_{in}| \geq \left(\sum_{\lambda \in \Lambda}P(\lambda)\right) \frac{q(q-1)}{2}\sqrt{|\F_\p|},$$
$$2|S_{out}| \geq \left(\sum_{\lambda \in \widehat{\Lambda}}P(\lambda)\right) \frac{q(q-1)}{2}\sqrt{|\F_\p|}.$$

The number of permutations in $S_{\deg(\p)}$ with no fixed points is $\lceil\deg(\p)!/e\rceil$ if $\deg(\p)$ is even and $\lfloor \deg(\p)!/e \rfloor$ otherwise. Thus $$\sum_{\lambda \in \widehat{\Lambda}}P(\lambda) \geq \frac{\lfloor \deg(\p)!/e \rfloor}{\deg(\p)! } \ \ \ and \ \ \ \sum_{\lambda \in \Lambda}P(\lambda) \geq  1 - \frac{\lceil\deg(\p)!/e\rceil}{\deg(\p) !}$$
and there exists positive constants $b_1,b_2$ such that $$\sum_{\lambda \in \Lambda}P(\lambda) \geq b_1 \ \ and \ \ \sum_{\lambda \in \widehat{\Lambda}}P(\lambda) \geq b_2.$$ 
Applying Lemma \ref{set_bound2} once each to $S_{in}$ and $S_{out}$, there exists positive constants $d_1$ and $d_2$ such that, $$\{(\phi/\p)|(a_{\phi,\p},\epsilon_{\phi,\p}) \in S_{in}\} \geq d_1|\F_\p| |\F_\p^\times|,$$$$\{(\phi/\p)|(a_{\phi,\p},\epsilon_{\phi,\p}) \in S_{out}\} \geq d_2|\F_\p| |\F_\p^\times|.$$
Since $g_\phi \in A$ and $\Delta_\phi \in A^\times$ are each chosen of degree less than $\deg(h)$ independently and uniformly at random and $\gcd(\Delta_\phi,h) = 1$, by the chinese remainder theorem $(g_\phi \mod \p, \Delta_\phi \mod \p)$ is distributed uniformly in $\F_{\p_i} \times \F_{\p_i}^\times$. Thus the probability that $Ann(\phi(\F_\p))$ has a linear factor is at least $d_1$ and the probability that $Ann(\phi(\F_\p))$ does not have a linear factor is at least $d_2$. Since the projection of a random $\alpha \in \phi(\F_\h)$ into $\phi(\F_\p)$ has order $Ann(\phi(\F_\p))$ with probability at least $(1-1/q) \geq 1/2$, the lemma follows.
\end{proof}
\begin{remark}\label{rem2} If $q < 5n\log n$, we may work over a finite extension $\F_q^\prime/\F_q$ such that $q^\prime \geq 3n \log n$ and by Lemma \ref{alg2_proof} be assured that the algorithm \ref{alg2} succeeds, however, we incur an extra factor of $n$ in the expected running time. If Conjecture \ref{smoothness_conjecture} is true, then Lemma \ref{alg2_proof} holds with only the requirements that $q$ is odd and $\sqrt{q} \geq 2n$. The assumption $\sqrt{q} \geq 2n$ may be made without loss in generality (Remark \ref{large_q}).
\end{remark}
\section*{Acknowledgement}
I thank Lior Bary-Soroker, Zeyo Guo, Ming-Deh Huang and Chris Umans for valuable discussions.


\begin{thebibliography}{99}
\bibitem[BBR14]{bbr} E. Bank, L. Bary-Soroker and L Rosenzweig, Prime polynomials in short intervals and in arithmetic progressions, http://arxiv.org/abs/1302.0625
\bibitem[BGJT]{bgjt} R. Barbulescu, P. Gaudry, A. Joux , E. Thome, ``A quasi-polynomial algorithm for discrete logarithm in finite fields of small characteristic", http://arxiv.org/abs/1306.4244
\bibitem[Bar12]{bs} L. Bary-Soroker, Irreducible values of polynomials, Adv. Math., 229 (2), 854-874 (2012). 
\bibitem[Ber67]{ber} E. R. Berlekamp, Factoring Polynomials Over Finite Fields, Bell System Tech. J., 46:1853-1849. 1967.
\bibitem[CM06]{cm} A. Cafure and G. Matera, Improved explicit estimates on the number of solutions of equations over a finite field, Finite Fields and Their Applications Volume 12, Issue 2, April 2006, Pages 155�185
\bibitem[Coh]{coh} S. D. Cohen, The Distribution of Polynomials over Finite Fields, Acta Arith., 17 (1970), 255-271.
\bibitem[CZ81]{cz} D. G. Cantor and H. Zassenhaus, A new algorithm for factoring polynomials over finite fields, Math. Comp., vol. 36, 587�592, 1981.
\bibitem[CP14]{cp} A. C. Cojocaru and M. Papikian, Drinfeld Modules, Frobenius Endomorphisms, and CM-Liftings, Int Math Res Notices (2014)
\bibitem[Cor99]{cor} G. Cornellison, Deligne's congruence and supersingular reduction of Drinfeld modules, Arch. Math. 72 (1999) 346-353.
\bibitem[Del74]{del} P. Deligne,  ``La conjecture de Weil. I", Publications Mathematiques de l'IHES (43): 273-307.
\bibitem[Deu41]{deu} M. Deuring, Die Typen der Multiplikatorenringe elliptischer Funktionenkorper, Abh. Math. Sem. Hamburg 14 (1941), 197�272. MR0005125.
\bibitem[Dor]{dor} D. R. Dorman, On singular moduli for rank 2 Drinfeld modules, Compositio Mathematica (1991) Volume: 80, Issue: 3, page 235-256.
\bibitem[Dri74]{dri} V. G. Drinfeld, Elliptic modules, Mat. Sb. (N.S.), 1974,	Volume 94(136), Number 4(8), Pages 594�627
\bibitem[Dri77]{dri1} V. G. Drinfeld, Elliptic modules. II, Mat. Sb. (N.S.), 102(144):2 (1977), 182-194.
\bibitem[vzGS92]{gs} J. von zur Gathen  and V. Shoup, Computing Frobenius maps and factoring polynomials, Comput. Complexity, vol. 2, 187-224, 1992.
\bibitem[Gek08]{gek} E-U Gekeler, Frobenius distributions of Drinfeld modules over finite fields, Trans. Amer. Math. Soc. 360 (2008), 1695-1721.
\bibitem[Gek91]{gek1} E-U Gekeler, On finite Drinfeld modules, Journal of Algebra, Volume 141, Issue 1, 1 August 1991, Pages 187-203.
\bibitem[Gos96]{gos} D. Goss, Basic Structures of Function Field Arithmetic, Ergeb. Math. Grenzgeb. (3), vol. 35, Springer, Berlin, 1996.
\bibitem[GNU15]{gnu} Z. Guo, A. Narayanan and C. Umans, Algebraic problems equivalent to beating the 3/2 exponent for polynomial factorization over finite fields, In preparation.
\bibitem[vdH04]{vdH} G. J. van der Heiden, Factoring polynomials over finite fields with Drinfeld modules, Math. Comp. 73 (2004), 317-322.
\bibitem[vdH04-1]{vdH1} G. J. van der Heiden, Addendum to ``Factoring polynomials over finite fields with Drinfeld modules",Math. Comp. 73(2004), Number 248, Page 2109
\bibitem[Jou]{jou} Antoine Joux. A new index calculus algorithm with complexity {{L}}$(1/4+o(1))$ in very small characteristic. Cryptology ePrint Archive, Report 2013/095, 2013.
\bibitem[KL94]{kl} E. Kaltofen and A. Lobo, Factoring high-degree polynomials by the black box Berlekamp algorithm, ISSAC '94 Proceedings of the international symposium on Symbolic and algebraic computation
Pages 90 - 98.
\bibitem[KS98]{ks} E. Kaltofen and V. Shoup, Subquadratic-time factoring of polynomials over finite fields, Math. Comput., 67(223):1179-1197, July 1998. 
\bibitem[KU08]{ku} K. Kedlaya and C. Umans, Fast modular composition in any characteristic, Proceedings of the 49th Annual IEEE Symposium on Foundations of Computer Science (FOCS). pages 146-155. 2008.
\bibitem[Len87]{len} H. W. Lenstra Jr, Factoring integers with elliptic curves, Annals of Mathematics 126 (3): 649�673. (1987).
\bibitem[Nar14]{nar} A. K. Narayanan, Computation of Class Groups and Residue Class Rings of Function Fields over Finite Fields, Ph.d Dissertation, Computer Science Department, University of Southern California 2014.
\bibitem[PP89]{pp} A. Panchishkin and I. Potemine, An algorithm for the factorization of polynomials using elliptic modules, In Proceedings of the Conference, Constructive methods and algorithms in number theory, p. 117. Mathematical Institute of AN BSSR, Minsk, 1989 (Russian).
\bibitem[Pol74]{pol} J. M. Pollard, Theorems of factorization and primality testing, Proceedings of the Cambridge Philosophical Society 76 (3): 521-528. (1974).
\bibitem[Sch95]{sch} R. Schoof: Counting Points on Elliptic Curves over Finite Fields. J. Theor. Nombres Bordeaux 7:219�254, 1995.
\bibitem[Sto01]{sto} A. Storjohann, Deterministic computation of the Frobenius form, In Proc. 42nd Annual Symp. Foundations of Comp. Sci., pages 368-377,  2001.
\bibitem[VS]{sal} G. D. Villa Salvador, Topics in the Theory of Algebraic Function Fields, 2006. Birkhauser, 2006.
\bibitem[Yu95]{yu} J-K Yu, Isogenies of Drinfeld modules over finite fields, J. Number Th. 54 (1995), 161-171.
\bibitem[Yun76]{yun} D. Y. Y. Yun, On square-free decomposition algorithms,  Proc. 1976 ACM Symp. on Symbolic and Algebraic Computation ISSAC 76.
\end{thebibliography}
\end{document}